\documentclass{paper-includes/llncs}

\newif\iffinal
\finaltrue
\newif\ifpagenumbers
\pagenumberstrue

\usepackage{amsthm}
\usepackage{amsfonts}
\usepackage{amsmath}
\usepackage{amssymb}
\usepackage{paper-includes/temporal}
\usepackage{graphicx}
\usepackage{caption}
\usepackage{xspace}
\usepackage{subcaption}
\usepackage{wrapfig}
\usepackage{stmaryrd} 
\usepackage{multirow}
\usepackage{doi}   
\usepackage{wrapfig}
\usepackage[cal=esstix]{mathalfa}
  \newcommand{\mc}{\mathcal}
\usepackage{xparse} 

\usepackage{enumitem}  
\setlist{topsep=3pt}

\iffinal
  \usepackage[finalold]{paper-includes/trackchanges} 
\else
  \usepackage[inline]{paper-includes/trackchanges}
\fi

\addeditor{A\!K}
\newcommand{\gray}[1]{{\color{black!70} #1}}	
\newcommand{\ak}[1]{\iffinal\else\gray{\noindent\scriptsize// #1}\fi}

\usepackage{hyperref}   
\hypersetup{
    colorlinks,
    linkcolor=black,
    citecolor=black,
}
\usepackage{bookmark}

\newcommand{\mailto}[1]{\href{mailto:#1}{\nolinkurl{#1}}}
\newcommand{\ack}[1]{\smallskip\noindent\small{\textbf{Acknowledgements.} #1}}

\newtheorem{obs}{Observation}

\newtheorem{cor}{Corollary}

\newcommand{\parbf}[1]{\smallskip\noindent {\bf #1.}}

\NewDocumentCommand \todo{g} { \IfValueTF{#1} {{\color{red}{\scriptsize \textbf{/#1/}}}} {{\color{red}{\small \tt todo}}} }

\newcommand\li{\begin{itemize}}
\newcommand\il{\end{itemize}}
\renewcommand{\-}{\item}
\newcommand\lo{\begin{enumerate}}
\newcommand\ol{\end{enumerate}}

\let\emptyset\varnothing

\newcommand\from\leftarrow
\newcommand\x\times
\newcommand\xx{\!\times\!}

\newcommand{\bbN}{\mathbb{N}}
\newcommand{\bbB}{\mathbb{B}}

\renewcommand{\|}{\ensuremath{\mid}} 

\newcommand{\cupdot}{\mathbin{\dot{\cup}}} 

\renewcommand{\iff}{\leftrightarrow}
\newcommand{\Iff}{\Leftrightarrow}
\newcommand{\impl}{\rightarrow}
\newcommand{\Impl}{\Rightarrow}

\newcommand{\lblTo}[1]{\stackrel{{#1}}{\longrightarrow}}
\newcommand{\lblToS}[1]{\stackrel{{#1}}{\rightarrow}}

\renewcommand{\true}{\textit{true}\xspace}
\renewcommand{\false}{\textit{false}\xspace}

\newcommand{\specialcellC}[2][c]{%
  \begin{tabular}[#1]{@{}c@{}}#2\end{tabular}}
\newcommand{\specialcellL}[2][c]{%
  \begin{tabular}[#1]{@{}l@{}}#2\end{tabular}}

\newcommand\G\always\xspace
\newcommand\F\eventually\xspace
\newcommand\W\weakuntil\xspace
\newcommand\U\until\xspace
\newcommand\R\releases\xspace
\newcommand\X\nextt\xspace

\newcommand\XX{\ensuremath{\nextt\!\nextt\xspace}}

\newcommand{\tpl}[1]{\left<#1\right>}

\iffinal
\renewcommand{\hl}[1]{#1}
\else
\fi

\NewDocumentCommand \UCWR{} {universal co-B\"uchi register automaton\xspace}

\NewDocumentCommand \NBWR{} {non-deterministic B\"uchi register automaton\xspace}

\NewDocumentCommand \gir{g} { \IfValueTF{#1} {{g_{\mc{ir}_{#1}}}} {{g_\mc{ir}}} }

\NewDocumentCommand \giv{g} { \IfValueTF{#1} {{g_{\mc{iv}_{#1}}}} {{g_\mc{iv}}} }
\NewDocumentCommand \gov{g} { \IfValueTF{#1} {{g_{{\mc ov}_{#1}}}} {{g_{\mc ov}}} }

\NewDocumentCommand \sr{g} { \IfValueTF{#1} {{a_{\mc r_{#1}}}} {{a_{\mc r}}} }
\NewDocumentCommand \sv{g} { \IfValueTF{#1} {{a_{\mc v_{#1}}}} {{a_{\mc v}}} }

\NewDocumentCommand \RTtuple{} {\tpl{I,O,\mc I, \mc O, \mc R, \mc d_0, S, s_0, \tau}}
\NewDocumentCommand \RAtuple{} {\tpl{P, \mc P, \mc R, \mc d_0, Q, q_0, \delta, F}}

\NewDocumentCommand \RAtuplePrime{} {\tpl{P, \mc P, \mc R, \mc d_0, Q', q'_0, \delta', F'}}
\NewDocumentCommand \RAtupleA{} {\tpl{P^A, \mc P, \mc R^A, \mc d_0, Q^A, q^A_0, \delta^A, F^A}}

\NewDocumentCommand \TALLtuple{} {\tpl{P^T, \mc P, \mc R^T, \mc d_0, Q^T, q^T_0, \delta^T, F^T}}
\NewDocumentCommand \VAtuple{} {\tpl{P, \mc P, \mc R, Q, q_0, \delta, F, E}}

\newcommand\ATV{AT_\bbB @ V_k}
\newcommand\ATW{AT_\bbB @ W}
\newcommand\hideATW{\textit{hide}_A(\ATW)}
\newcommand\Tall{T^\textit{\!all}}
\newcommand\AT{A \ox \Tall}
\newcommand\Asgn{\textit{Asgn}}

\DeclareMathOperator*{\ox}{\text{\raisebox{0.25ex}{\scalebox{0.7}{$\otimes$}}}}

\usepackage[svgnames]{xcolor}
\usepackage{tikz}
\usetikzlibrary{shapes,backgrounds}
\usetikzlibrary{decorations.markings}
\usetikzlibrary{shapes.geometric}
\usetikzlibrary{shapes.misc}
\usetikzlibrary{arrows,automata}
\usetikzlibrary{arrows.meta}
\usetikzlibrary{decorations.pathmorphing}

\pgfdeclarelayer{edgelayer}
\pgfdeclarelayer{nodelayer}
\pgfsetlayers{edgelayer,nodelayer,main}

\tikzstyle{none}=[inner sep=0pt]
\tikzset{initial text={}}

\tikzstyle{rn}=[circle,fill=White,draw=Red,minimum size=0.5cm,inner sep=0pt]
\tikzstyle{gn}=[circle,fill=White,draw=Green,minimum size=0.5cm,inner sep=0pt]
\tikzstyle{yn}=[circle,fill=White,draw=Yellow,minimum size=0.5cm,inner sep=0pt]
\tikzstyle{wn}=[circle,fill=White,draw=Black,minimum size=0.5cm,inner sep=0pt]
\tikzstyle{invisible}=[circle,fill=White,draw=White,inner sep=0pt]
\tikzstyle{textual}=[rectangle,font=\footnotesize]
\tikzstyle{dot}=[circle,fill=Black,draw=Black,inner sep=0pt,minimum size=0.5mm]
\tikzstyle{uptriangle}=[regular polygon,regular polygon sides=3,shape border rotate=0,fill=Black,draw=Black,inner sep=0pt,minimum size=1mm]
\tikzstyle{text rounded}=[rectangle,rounded corners=2.2mm,fill=White,draw=Black]
\tikzstyle{text rectangle}=[rectangle,fill=White,draw=Black]
\tikzstyle{text ellipse}=[ellipse,fill=White,draw=Black,inner sep=1pt]

\tikzstyle{simple}=[-,draw=Black]
\tikzstyle{arrow}=[->,>={Stealth},draw=Black]
\tikzstyle{dashed arrow}=[->,draw=Black,dashed]
\tikzstyle{tick}=[-,draw=Black,postaction={decorate},decoration={markings,mark=at position .5 with {\draw (0,-0.1) -- (0,0.1);}}]
\tikzstyle{ga}=[->,draw=Green]
\tikzstyle{pa}=[->,draw=DeepPink]
\tikzstyle{red}=[-,draw=Red]
\tikzstyle{pink}=[-,draw=DeepPink]
\tikzstyle{blue}=[-,draw=Blue]
\tikzstyle{green}=[-,draw=Green]
\tikzstyle{yellow}=[-,draw=DarkGoldenrod]

\iffinal
\else
\fi

\title{Bounded Synthesis of Register Transducers \texorpdfstring{\vspace{-3mm}}{}}

\author{%
  Ayrat Khalimov$^1$,
  Benedikt Maderbacher$^2$,
  Roderick Bloem$^2$%
  \institute{%
    \vspace{-2mm}
    $^{2}$ Graz University of Technology, Austria\\
    $^{1}$ Hebrew University, Israel%
  }
}

\begin{document}
\footnotetext[0]{This is a full version of our ATVA'18 paper.}

\ifpagenumbers
  \pagestyle{plain}
\fi

\setlength\abovedisplayskip{4pt}
\setlength\belowdisplayskip{4pt}

\setcounter{tocdepth}{3}

\maketitle
\vspace{-7mm}

\begin{abstract}
  Reactive synthesis aims at automatic construction of systems from their behavioural specifications.
  The research mostly focuses on synthesis of systems dealing with Boolean signals.
  But real-life systems are often described using bit-vectors, integers, etc.
  Bit-blasting would make such systems unreadable, hit synthesis scalability, and
  is not possible for infinite data-domains.
  One step closer to real-life systems are register transducers~\cite{Kaminski}:
  they can store data-input into registers and later output the content of a register,
  but they do not directly depend on the data-input, only on its comparison with the registers.
  Previously~\cite{ehlers-register-machines} it was proven that synthesis of register transducers
  from register automata is undecidable,
  but there the authors considered transducers equipped with the unbounded queue of registers.
  First, we prove the problem becomes decidable if bound the number of registers in transducers,
  by reducing the problem to standard synthesis of Boolean systems.
  Second, we show how to use quantified temporal logic, instead of automata, for specifications.
\end{abstract}
\vspace{-8mm}

\section{Introduction} \label{sec:intro}
 
  Reactive synthesis~\cite{Church63} frees hardware and software developers
  from tedious and error-prune coding work.
  Instead, the developer specifies the desired behaviour of a system,
  and a synthesizer produces the actual code.
  The research in reactive synthesis is mostly focused on synthesis of transducers
  dealing with \emph{Boolean} inputs and outputs.
  However,
  most programs and hardware designs use not only Booleans,
  but also bit-vectors, integers, reals.
  Bit-blasting into Booleans makes synthesized programs unreadable and hinders the synthesis scalability.
  
  One step closer to real-life systems are register transducers~\cite{Kaminski}.
  Such transducers are equipped with registers;
  they can read the data-input from an infinite domain;
  they can store the data-input into a register and later output it;
  they do not depend on the exact data-input value, but on its comparison with the registers.
  Thus, a transition of a register transducer can say
  ``in state $q$:
    if the data-input not equals to register \#1, then
    output the value of register \#1, store the data-input into register \#2, and go into state $q'$''.
  Examples of a register transducer and automaton are in Figures~\ref{fig:rex:machine} and \ref{fig:rex:automaton}.
  
  In \cite{ehlers-register-machines}, the authors introduced the problem of synthesis of register transducers.
  But their transducers are equipped with an \emph{unbounded queue} of registers:
  they can push the data-input into the queue, and later compare the data-input with the values in the queue.
  For specifications, the authors use register automata with a fixed number of registers (thus, no queue).
  The authors show that the synthesis problem is undecidable;
  the proof relies on unboundedness of the queue.
  
  We prove the problem becomes decidable if bound the number of registers in transducers.
  Namely, we reduce synthesis of $k$-register transducers wrt.\ register automata
  to synthesis of Boolean transducers wrt.\ Boolean automata,
  i.e., to standard synthesis.
  The reduction relies on two ideas.
 
  The first (folklore) idea is:
  instead of tracking the exact register values and data-inputs,
  track only the \emph{equivalences} between register values and the data-input.
  The second idea is:
  instead of checking automaton non-emptiness,
  we check automaton non-emptiness \emph{modulo words of $k$-register transducers}.
  Every such word can be enhanced with assignment actions of the transducer that resulted in producing the word.
 
 
  In the second part,
  we suggest a temporal logic that ``works well'' with our approach.
  Among several logics suitable to the context of infinite data~\cite{Wolper-data-indep,grumberg2012model,LTLfreeze,Demri:2012:TLR:2400051.2400054},
  we have chosen IPTL~\cite{Wolper-data-indep} (called VLTL in \cite{grumberg2012model}),
  because of its naturalness.
  Using this logic, we can state properties like
  $\forall \mc d \in \mc D: \G(\mc i = \mc d \impl \F (\mc o = \mc d))$:
  ``every data-value appearing on the input eventually appears on the output''.
  We show how to convert a formula in this logic into a register automaton (in incomplete way; there can be no complete way)
  that can be used by our synthesis approach.

\section{Definitions} \label{sec:defs}

  Fix a \emph{data-domain} $\mc D$ throughout the paper,
  which is an infinite set of elements (\emph{data-values}).
  Calligraphic writing like $\mc i$, $\mc o$, $\mc d$, $\mc r$ denotes data-variables or objects closely related to them.
  Sets of such objects are also written in calligraphic, like $\mc D$, $\mc R$, $\mc P$, etc.
  Define $\bbN = \{1,2,...\}$, $\bbN_0 = \{0,1,2,...\}$, $[k] = \{1,...,k\}$ for $k \in \bbN$;
  $\bbB = \{\true, \false\}$, and we often use the subscripted variants, $\bbB_\mc i = \bbB_\mc o = \bbB$,
  to clarify when $\bbB$ is related to object $\mc i$ or $\mc o$.
  For an automaton $A$, let $L(A)$ denote the set of its accepting words.

  \subsection{Register Automata}

    A register automaton works on words from $(2^P \x \mc D^{\mc P})^\omega$,
    where $P$ is a set of Boolean signals and $\mc P$ is a set of data-signals.
    To simplify the presentation, we assume there are only two data-signals ($\mc P = \{\mc i, \mc o\}$),
    which makes the words to be from $(2^{P} \x \mc D^2)^\omega$.
    When reading a word,
    a register automaton can store the value of data-signal $\mc i$ into its registers.
    Later it can compare the content of its registers with the current value of $\mc i$.
    Register automata do not depend on actual data-values---only on the comparison with the register values.
    Below is a formal definition.

    A \emph{(universal co-B\"uchi/non-deterministic B\"uchi) word automaton with $k$ registers} is a tuple $A = \RAtuple$, where
    \li
    \- $P$ is a set of \emph{Boolean signals};

    \- $\mc P = \{\mc i, \mc o\}$ is a set of \emph{data-signals};

    \- $\mc R = \{\mc r_1, ..., \mc r_k\}$ is a set of \emph{registers};

    \- $\mc d_0 \in \mc D$ is an \emph{initial data-value} for every register;

    \- $Q$ is the set of \emph{states} and $q_0 \in Q$ is an \emph{initial state};

    \- $F \subseteq Q$ is a set of \emph{accepting states};

    \- $\delta:
        Q \x
        2^{P} \x 
        \bbB_\mc i^k \x
        \bbB_\mc o^k
        \to
        2^{\bbB^k \!\x Q}$
       is a \emph{transition function}.
       Intuitively,
       in a state, an automaton reads
       a finite letter from $2^P$ (which describes all Boolean signals whose current value is true)
       and a data-letter from $D^2$ (a data-value for $\mc i$ and a data-value for $\mc o$).
       Then the automaton compares the data-letter with the content of the registers.
       Depending on this comparison (component $\bbB_\mc i^k \x \bbB_\mc o^k$, called \emph{guard}),
       the automaton transits into several (for universal automaton) or one of (for non-deterministic automaton) successor states,
       and for each successor state,
       stores the value of data-signal $\mc i$ into one, several, or none of the registers
       (defined by component $\bbB^k$, called \emph{assignment} or \emph{store}).
    \il
    An example of a register automaton is in Figure~\ref{fig:rex:automaton}.

    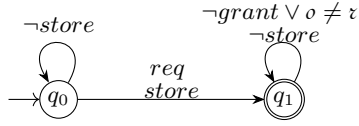
\begin{figure}[tb]
      \centering
      \begin{tikzpicture}
	\begin{pgfonlayer}{nodelayer}
		\node [initial, style=wn] (0) at (-1, 1.25) {$q_0$};
		\node [style=wn, double] (1) at (2, 1.25) {$q_1$};
		\node [style=textual] (2) at (-1, 2.25) {$\neg store$};
		\node [style=textual] (3) at (0.5, 1.5) {\specialcellC{$req$ \vspace{-0.8mm}\\ $store$}};
		\node [style=textual] (4) at (2, 2.25) {\specialcellC{$\neg grant \lor \mc o \neq \mc r$  \vspace{-0.8mm} \\ $\neg store$ }};
	\end{pgfonlayer}
	\begin{pgfonlayer}{edgelayer}
		\draw [style=arrow] (0) to (1);
		\draw [style=arrow, in=120, out=60, loop] (1) to node[double]{} ();
		\draw [style=arrow, in=120, out=60, loop] (0) to ();
	\end{pgfonlayer}
\end{tikzpicture}
      \vspace{-3mm}
      \caption{A universal co-B\"uchi 1-register automaton:
        $P=\{req,grant\}$, $\mc R = \{\mc r\}$, $F = \{q_1\}$.
        The labels $\neg store$ and $store$ have a special meaning:
        $store$ means that the automaton stores the value of data-input $\mc i$ into register $\mc r$;
        $\neg store$ means it does not.
        The expression $\mc o \neq \mc r$ means that the component $\bbB_\mc o$ of the transition is $\false$.
        For guards and Boolean signals, the labeling is symbolic.
        Formally, the set of transitions is
        $
        \big\{(q_0, p, b_\mc i, b_\mc o, \false, q_0): (b_\mc i,b_\mc o) \in \bbB^2, p \in 2^P\big\}
        \cup
        \big\{(q_0, p, b_\mc i, b_\mc o, \true, q_1): (b_\mc i,b_\mc o) \in \bbB^2, req \in p \in 2^P\big\}
        \cup
        \big\{(q_1, p, b_\mc i, b_\mc o, \false, q_1): (b_\mc i,b_\mc o) \in \bbB^2, p \in 2^P, grant \not\in p \lor b_\mc o=\false\big\}$.%
      }
      \label{fig:rex:automaton}
    \end{figure}

    A \emph{configuration} is a tuple $(q, \mc {\bar d}) \in Q \x \mc D^k$,
    and $(q_0, \mc d_0^k)$ is \emph{initial}.
    A \emph{path} \label{page:atm-path} is an infinite sequence
    $(q_0, \mc {\bar d}_0) \lblTo{(l_0, \mc i_0, \mc o_0, \bar a_0)} (q_1, \mc {\bar d}_1) \lblTo{(l_1, \mc i_1, \mc o_1, \bar a_1)}...$
    such that for every $j \in \bbN_0$:
    \li
    \- $q_j \in Q$, $\mc {\bar d}_j \in \mc D^k$, $l_j \in 2^P$, $\mc i_j \in \mc D$, $\mc o_j \in \mc D$, and $\bar a_j \in \bbB^k$;

    \- $(q_{j+1}, \bar a_j) \in
        \delta\big(q_j, l_j,
          \mc i_j = \mc {\bar d}_j[1], ..., \mc i_j = \mc {\bar d}_j[k],
          \mc o_j = \mc {\bar d}_j[1], ..., \mc o_j = \mc {\bar d}_j[k]\big)$;

    \- $\mc {\bar d}_0 = \mc d_0^k$; and

    \- \vspace{-3mm} for every $n \in [k]$:
       $\mc {\bar d}_{j+1}[n] =
       \begin{cases}
         \mc i_j &\textit{if } \bar a_j[n] = \true,\\
         \mc {\bar d}_j[n] &\textit{otherwise}.
       \end{cases}$
    \il
    
    \noindent
    Let $\Sigma = 2^{P} \x \mc D^2$.
    A \emph{word} is a sequence from $\Sigma^\omega$.
    A word is \emph{accepted} by a \UCWR iff every path%
    ---whose projection into $\Sigma$ equals to the word---%
    does not visit a state from $F$ infinitely often;
    otherwise the word is \emph{rejected}.
    A word is \emph{accepted} by a \NBWR iff there is a path%
    ---whose projection into $\Sigma$ equals to the word---%
    that visits a state from $F$ infinitely often;
    otherwise the word is \emph{rejected}.
    For example,
    the \UCWR in Figure~\ref{fig:rex:automaton} accepts the word \\
    $(\{req\} , 5_\mc i,*_\mc o) (\{req, grant\}, 6_\mc i,5_\mc o) (\{grant\}, *_\mc i,6_\mc o) (\emptyset, *_\mc i, *_\mc o)^\omega$,
    where $\mc D = \bbN_0$, we write subscripts $\mc i$ and $\mc o$ for clarity,
    and $*$ is anything from $\mc D$ (not necessary the same).
    The automaton describes the words where every $req$
    is followed by $grant$ with the data-value of $\mc o$ being equal to the data-value of $\mc i$ at the moment of the request.
    Such words can be described by a formula $\forall \mc d \in \mc D: \G\big(req \land \mc i = \mc d \impl \X\F (grant \land \mc o = \mc d)\big)$,
    but we postpone the discussion of logic until Section~\ref{sec:logic}.

  \subsection{Register Transducers}

    Register transducers is an extension of standard transducers (Mealy machines) to an infinite domain.
    A register transducer can store the input data-value into its registers.
    It can only output the data-value that is currently stored in one of its registers.
    Similarly to register automata,
    the transitions of register transducers depend on the comparison of the data-input with the registers,
    but not on the actual data-values.
    Let us define register transducers formally.
 
    A \emph{$k$-register transducer} is a tuple $T = \RTtuple$ where:
    \li
    \- $I$ and $O$ are sets of Boolean signals, called \emph{Boolean inputs} and \emph{outputs};

    \- $\mc I$ and $\mc O$ are sets of data-signals, called \emph{data-inputs} and \emph{data-outputs};
       we assume that $\mc I = \{\mc i\}$ and $\mc O = \{\mc o\}$.

    \- $S$ is a (finite or infinite) set of \emph{states} and $s_0\in S$ is \emph{initial};

    \- $\mc R = \{\mc r_1,...,\mc r_k\}$ is a set of \emph{registers};

    \- $\mc d_0 \in \mc D$ is an \emph{initial data-value} for every register;

    \- $\tau: S \x 2^I \x \bbB_\mc i^k  \to  (2^O \x \mc [k] \x \bbB^k \x S)$ is a \emph{transition function}.
       Intuitively, from a state the transducer reads
       the values of the Boolean inputs (component $2^I$) and
       compares the content of the registers with the data-value of $\mc i$
       (component $\bbB^k_\mc i$, called \emph{guard}).
       Depending on that information,
       the transducer transits into a unique successor state (component $S$),
       stores the data-value of $\mc i$ into one, several, or none of the registers (component $\bbB^k$, called \emph{assignment} or \emph{store}),
       outputs a value for each Boolean output (component $2^O$),
       and outputs a data-value stored in one of the registers (component $[k]$).

    \il
    Figure~\ref{fig:rex:machine} shows an example of a register transducer.

    \begin{figure}[bt]
      \centering
      \begin{tikzpicture}
	\begin{pgfonlayer}{nodelayer}
		\node [style=wn, initial] (0) at (-1, 1.25) {$s_0$};
		\node [style=wn] (1) at (2, 1.25) {$s_1$};
		\node [yshift=1.5mm, style=textual] (2) at (-2.25, 1.75) {\specialcellC{$\neg req/\neg grant$  \vspace{-0.8mm} \\ $\neg store$ }};
		\node [style=textual] (3) at (0.5, 1.75) {\specialcellC{$req/\neg grant$\vspace{-0.5mm} \\ $store$}};
		\node [style=textual] (4) at (3, 1.75) {\specialcellC{$req/grant$ \vspace{-0.5mm} \\ $store$}};
		\node [yshift=0mm, style=textual] (5) at (0.5, 0.75) {\specialcellC{$\neg req/grant$   \vspace{-0.8mm}  \\ $\neg store$}};
	\end{pgfonlayer}
	\begin{pgfonlayer}{edgelayer}
		\draw [style=arrow, bend left=15, looseness=0.50] (0) to (1);
		\draw [style=arrow, in=120, out=60, loop] (1) to node{} ();
		\draw [style=arrow, in=120, out=60, loop] (0) to ();
		\draw [style=arrow, bend left=15, looseness=0.50] (1) to (0);
	\end{pgfonlayer}
\end{tikzpicture}
      \vspace{-3mm}
      \caption{A 1-register transducer: $I = \{req\}, O = \{grant\}$, $\mc R = \{\mc r\}$.
        The meaning of $store$ and $\neg store$ is as in the previous figure.
        The labeling wrt.\ guards and Boolean signals is symbolic.
        The transducer always outputs the value of its only register (not shown).
        Formally, the set of transitions is
        $
        \big\{(s_0, \emptyset, b_\mc i, \emptyset, 1, \false, s_0): b_\mc i \in \bbB \big\}
        \cup
        \big\{(s_0, \{req\}, b_\mc i, \emptyset, 1, \true, s_1): b_\mc i \in \bbB \big\}
        \cup
        \big\{\!(s_1, \{req\!\}, b_\mc i, \{grant\}, 1, \true, s_1)\!\!:\!\!b_\mc i\!\in\!\bbB \big\}\!
        \!\cup\!
        \big\{\!(s_1, \emptyset, b_\mc i, \{grant\}, 1, \false, s_0)\!:\!b_\mc i\!\in\!\bbB \big\}
        $.%
        }
      \label{fig:rex:machine}
    \end{figure}
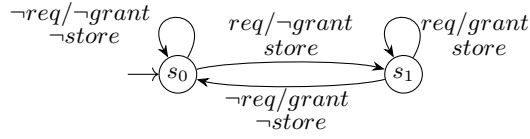

    A \emph{configuration} is a tuple $(s, \mc {\bar d}) \in Q \x \mc D^k$;
    $(s_0,\mc d_0^k)$ is called \emph{initial}.
    A \emph{path} is a sequence
    $(s_0,\mc {\bar d}_0) \lblTo{(i_0,o_0,\mc i_0,\mc o_0, \bar a_0)} (s_1,\mc {\bar d}_1) \lblTo{(i_1,o_1,\mc i_1,\mc o_1, \bar a_1)} ...$ 
    where for every $j \in \bbN_0$:
    \li
    \- $s_j \in S$, $\mc {\bar d}_j \in \mc D^k$, $i_j \in 2^I$, $o_j \in 2^O$, $\mc i_j \in \mc D$, $\mc o_j \in \mc D$, $\bar a_j \in \bbB^k$;

    \- let $(out, \mc{out}, store, succ) = \tau(s_j, i_j, \mc i_j = \mc {\bar d}_j[1],  ..., \mc i_j = \mc {\bar d}_j[k])$.
       Then:

    \- $s_{j+1} = succ$;

    \- $\bar a_j = store$;

    \- $\mc o_j = \mc {\bar d}_j[\mc{out}]$;

    \- $o_j = out$;

    \- $\mc {\bar d}_0 = \mc d_0^k$; and
    
    \- \vspace{-3mm} for every $n \in [k]$:
       $
       \mc {\bar d}_{j+1}[n] =
       \begin{cases}
         \mc i_j &\textit{if } \bar a_j[n] = \true,\\
         \mc {\bar d}_j[n] &\textit{otherwise}.
       \end{cases}
       $
    \il
    Notice that a value of the data-output refers to the current register values, not the updated ones.
    I.e., outputting a data-value happens before storing.

    For example, a path of the register transducer in Figure~\ref{fig:rex:machine} can start with
    $
    (s_0, 0)\! \lblTo{(\{req\}, \emptyset, 5_\mc i, 0_\mc o, \true)}\!
    (s_1, 5)\! \lblTo{(\{req\}, \{grant\}, 6_\mc i, 5_\mc o, \true)}\!
    (s_1, 6)\! \lblTo{(\emptyset, \{grant\}, 4_\mc i, 6_\mc o, \false)}\!
    (s_0, 6),
    $
    where we assumed that $\mc D = \bbN_0$, $\mc d_0 = 0$, and the subscripts $\mc i$ and $\mc o$ are for clarity.

    A \emph{word} is a projection of a transducer path into $2^{I\cup O} \x \mc D^2$.
    A register transducer \emph{satisfies} a register automaton $A$, written $T \models A$,
    iff all transducer words are accepted by the automaton.
    For example, the register transducer from Figure~\ref{fig:rex:machine} satisfies the automaton from Figure~\ref{fig:rex:automaton}.

  \subsection{Synthesis Problem}

    In this section, we define the model checking problem, bounded, and unbounded-but-finite synthesis problems.
    All the problems take as input a universal register automaton:
    one argument in favour of universal rather than non-deterministic automata is that the property
    ``every data-request is eventually data-granted''%
    can be expressed with a universal automaton, but not with a nondeterministic automaton.

    \parbf{Model checking and cutoffs}
    The \emph{model-checking problem} is:

    \li
    \- Given: a register transducer $T$, a universal co-B\"uchi register automaton $A$.
    \- Return: ``yes'' if $T \models A$, otherwise ``no''.
    \il
    \label{page:cutoff}
    The model-checking problem is decidable, which follows from the following.
    Kaminski and Francez~\cite[Prop.4]{Kaminski} proved the following \emph{cutoff result} (adapted to our notions):
    if a data-word over an infinite domain $\mc D$ is accepted by a non-deterministic B\"uchi $k$-register automaton,
    then there is an accepting data-word over a finite domain $\mc D_{k+1}$ of size $k+1$.
    (Actually, their result is for words of finite length, but can be extended to infinite words.)
    Further, if we look at a given universal co-B\"uchi $k_A$-register automaton $A$ as being non-deterministic B\"uchi $\widetilde{A}$,
    then $L(\widetilde{A}) = \overline{L(A)}$, i.e., it describes the error words.
    To do model checking, as usual,
    (1) build the product of the $\widetilde{A}$ and a given $k_T$-register transducer $T$, then
    (2) check its emptiness and return ``the transducer is correct'' iff the product is empty.
    The product is easy to build, this is an easy extension of the standard product construction,
    we note only that it is a non-deterministic B\"uchi $(k_A+k_T)$-register automaton.
    Finally, to check emptiness of the product we can use the cutoff result,
    namely, restrict the data-domain to have $(k_A+k_T+1)$ data-values.
    This reduces product emptiness to standard emptiness of register-less automata.

    The case of deterministic Rabin register automata and transducers with more than single data-input and data-output
    was studied in~\cite{10.1007/3-540-44618-4_41}, but the proof idea is similar.
    \ak{add more pointers?}

    In this paper we focus on the synthesis problem defined below.

    \parbf{Synthesis}
    The \emph{bounded synthesis problem} is:
    \li
    \- Given:
       a register-transducer interface
       (the number of registers $k_T$,
        Boolean and data-inputs,
        Boolean and data-outputs), a universal co-B\"uchi register automaton $A$.
    \- Return: a $k_T$-register transducer $T$ of a given interface such that $T \models A$,
               otherwise ``unrealizable''.
    \il
    If the number of registers $k_T$ is not given
    (thus we ask to find any such $k_T$ which makes the problem realizable, or return ``unrealizable'' if no such $k_T$ exists),
    then we get the (finite but unbounded) \emph{synthesis problem}.

    A related synthesis problem (let us call it ``infinite synthesis problem'') was studied in~\cite{ehlers-register-machines},
    but for a slightly different model of register transducers.
    There, the transducers operate an unbounded queue of registers (thus, it may use an infinite number of registers).
    They prove the infinite synthesis problem is undecidable
    and suggest an incomplete synthesis approach.\ak{how does it work?}

    In the next sections,
    we show that the bounded synthesis problem is decidable,
    and suggest an approach that reduces it to the synthesis problem of register-less transducers wrt.\ register-less automata.
    The (unbounded) synthesis problem is left open.

    But before proceeding to our solution,
    let us remark why the cutoff result does not immediately give a complete synthesis procedure.
    
    \begin{remark}[Cutoffs and synthesis]
      The cutoff result makes the data-domain finite,
      so let the values of the registers be part of the transducer states.
      Then a transducer has to satisfy the three conditions below,
      where condition (3) explains why the cutoff does not work with this naive approach.
      \li
      \-[(1)] ``The register values are updated according to transducer store actions.''

        Introduce new Boolean outputs describing the current values of the transducer registers,
        and new Boolean outputs describing the store action.
        Then it is easy to encode the above requirement using a register-less automaton.

      \-[(2)] ``The value of the data-output always equals the value of one of the registers.''

        With the Boolean outputs introduced in item (1),
        this can be easily encoded using a register-less automaton.
      
      \-[(3)] ``The transitions depend on the guard, but not on the value of data-input.''

        When considered alone, this requirement can be implemented using the partial-information synthesis approach~\cite{KV97c},
        where we search for a transducer that can access the guard, but not the actual value of data-input.
        But the partial-information synthesis approach does not allow for having partial information for transitions (needed to implement item (3)),
        yet full information for outputs (needed to implement items (1) and (2)).
      \il
      Nevertheless, with the cutoff it is easy to get an \emph{incomplete} synthesis approach with SMT-based bounded synthesis~\cite{BS}
      that allows you to fine-tune transition and output functions dependencies.
    \end{remark}

\section{Solving the Bounded Synthesis Problem} \label{sec:synthesis-solution}

  Our approach is 5 points long.

  (\textbf 1) We start by defining a Boolean associate $A_\bbB$ of a \UCWR $A$,
  which is a standard register-less universal co-B\"uchi automaton derived from the description of $A$.
  Of course, we cannot directly use the Boolean associate $A_\bbB$ to answer questions about $A$,
  because $A_\bbB$ lacks the semantics of $A$.
  We also define a Boolean associate $T_\bbB$ for every register transducer $T$.
  In the end, we will synthesize $T_\bbB$ that satisfies a certain register-less automaton.
  For examples of such associates,
  look at the automaton and transducer on Figures~\ref{fig:rex:automaton}~and~\ref{fig:rex:machine}
  as being standard, register-less, where $store$ is a Boolean signal and has no special meaning.
  (\textbf 2) We introduce a verifier automaton $V$,
  which tracks the equivalences between the registers $\mc R^A$ of $A$:
  two registers fall into the same equivalence class iff they hold the same data-value.
  The automaton $A_\bbB @ V$ is $A_\bbB$ enhanced with this equivalence-class information.
  It has enough information to answer the questions like ``does $A$ have a \emph{rejecting} word?'' and model checking wrt.\ $A$.
  This is because every Boolean path of $A_\bbB @ V$ corresponds to some data-path in $A$, and vice versa
  (which was not the case for $A_\bbB$ and $A$).
  But $A_\bbB @ V$ is not suited for synthesis%
  ---we cannot synthesize from $A_\bbB @ V$---%
  for one of the two reasons:
  either we would have to allow the transducers to control the store actions of $A$, which brings unsoundness,
  or we would have to allow the environment to provide the input guards that do not correspond to any data-value,
  which brings incompleteness.
  (\textbf 3) We add $k_T$ fresh registers $\mc R^T$ to $A$ that will be controlled by a transducer.
  To this end, we define the automaton $\Tall$:
  it reads data-words \emph{enhanced with store information of a transducer},
  and filters out data-words that do not belong to any of the $k_T$-register transducers
  (e.g., data-words that have a value for $\mc o$ that was not seen before on $\mc i$).
  We define $A \ox \Tall$, whose language is $L(A) \cap L(\Tall)$\footnotemark[1].
  (\textbf 4) We enhance the Boolean associate $(A \ox \Tall)_\bbB$ of $A \ox \Tall$ with information about equivalences between
  the registers $\mc R^T$ \emph{and} $\mc R^A$;
  the resulting automaton is called $(A \ox \Tall)_\bbB @ W$,
  where $W$ is a verifier similar to $V$ but tailored towards synthesis.
  (\textbf 5) Finally, we hide the information that should not be visible to a transducer,
  namely information related to the automaton registers $\mc R^A$.
  The resulting automaton is called $H = hide_A((A \ox \Tall)_\bbB @ W)$ and it is such that
  $\exists T: T \models A$ iff $\exists T_\bbB: T_\bbB \models H$.
  Furthermore, $H$, when viewed as a register automaton, is determinizable, and $L(H) \subseteq L(A)$\footnotemark[1].
  \footnotetext[1]{Actually, their alphabets differ, so this statement assumes $A$ with extended alphabet.}

  \subsection{Boolean Associates of Register Automata and Transducers} \label{sec:associates}

    The transition functions of $k$-register automata do not contain any infinite objects---data-values appear only in the semantics.
    Let us define Boolean associates of register automata and transducers.

    Given a $k$-register automaton $A = \RAtuple$,
    let \emph{Boolean automaton} $A_\bbB = \tpl{P_\bbB, Q, q_0, \delta_\bbB, F}$ be a standard register-less automaton
    where:
    \li
    \- let
       $G_\mc i = \{g_{\mc{ir}_1}, ..., g_{\mc{ir}_k}\}$,
       $G_\mc o = \{g_{\mc{or}_1}, ..., g_{\mc{or}_k}\}$,
       $Asgn = \{a_{\mc{r}_1}, ..., a_{\mc{r}_k}\}$.
       Then:

    \- $P_\bbB = P \cup G_\mc i \cup G_\mc o \cup Asgn$,

    \- $\delta_\bbB: Q \x 2^{P_\bbB} \to 2^Q$ contains
       $(q, l\cup g_\mc i \cup g_\mc o \cup a, q') \in \delta_\bbB$ iff $(q, l, \bar b_\mc i, \bar b_\mc o, \bar a, q') \in \delta$,
       where
       $l \in 2^P$, $g_\mc i \in 2^{G_\mc i}$, $g_\mc o \in 2^{G_\mc o}$, $a \in 2^{Asgn}$,
       $\bar b_\mc i = (g_{\mc{ir}_1} \in g_\mc i, ..., g_{\mc{ir}_k} \in g_\mc i) \in \bbB^k$,
       $\bar b_\mc o = (g_{\mc{or}_1} \in g_\mc o, ..., g_{\mc{or}_k} \in g_\mc o) \in \bbB^k$,
       $\bar a = (a_{\mc{r}_1} \in a, ..., a_{\mc{r}_k} \in a) \in \bbB^k$.
       Informally, we take the assignment component (on the right side) of $\delta$ and move it to the left side of $\delta_\bbB$,
       and introduce new Boolean signals to describe the Boolean components.
    \il
    For convenience,
    we say that a letter $g_\mc i \in 2^{G_\mc i}$
    \emph{encodes} the guard
    $(g_{\mc{ir}_1} \in g_\mc i, ..., g_{\mc{ir}_k} \in g_\mc i) \in \bbB^k$,
    and vice versa;
    similarly for a letter from $2^{G_\mc o}$ and $2^{\Asgn}$.

    A \emph{Boolean path} is an infinite sequence
    $q_0 \lblTo{l_0 \cup g_\mc i{}_0 \cup g_\mc o{}_0 \cup a_0} q_1 \lblTo{l_1\cup g_\mc i{}_1\cup g_\mc o{}_1\cup a_1} ...$
    from $(Q \x 2^{P_\bbB})^\omega$ that satisfies $\delta_\bbB$.
    When necessary to distinguish paths of register automata (which are in $(Q \x \mc D^k \x 2^P \x \mc D^2)^\omega$)
    from Boolean paths,
    we call the former \emph{data-paths}.
    A data-path $(q_0, \bar{\mc d}_0) \lblTo{(l_0, \mc i_0, \mc o_0, \bar a_0)} (q_1, \bar{\mc d}_1) \lblTo{(l_1, \mc i_1, \mc o_1, \bar a_1)}...$
    \emph{corresponds} to a Boolean path
    $q_0 \lblTo{l_0\cup g_\mc i{}_0\cup g_\mc o{}_0 \cup a_0} q_1 \lblTo{l_1\cup g_\mc i{}_1\cup g_\mc o{}_1\cup a_1}...$
    where
    $g_\mc i{}_j$ encodes the guard $(\mc i_j = \bar {\mc d}_j[1], ..., \mc i_j = \bar {\mc d}_j[k])$,
    $g_\mc o{}_j$ encodes the guard $(\mc o_j = \bar {\mc d}_j[1], ..., \mc o_j = \bar {\mc d}_j[k])$,
    and $a_j \in 2^{\Asgn}$ encodes $\bar a_j \in \bbB^k$,
    for $j \in \bbN_0$.
    From the definition of paths of register automata on page~\pageref{page:atm-path},
    it follows that for every path of a register automaton,
    there exists a path in the associated Boolean automaton to which the data-path corresponds.
    Consider the reverse direction,
    where we say that a Boolean path \emph{corresponds} to a data-path iff the data-path corresponds to it.
    The reverse direction does not necessarily hold:
    there is a register automaton $A$ (e.g., with 2 registers)
    where some Boolean paths of $A_\bbB$ do not have a corresponding data-path in $A$.
    This is because the letters of a Boolean path can describe contradictory guards.
    For example, let a transition in a Boolean path have $\bar a = (\true,\true)$,
    meaning that in a data-path the value of data-input is stored into the registers $\mc r_1$ and $\mc r_2$.
    Hence, in the next transition of the data-path,
    $\mc i = \mc r_1 \Iff \mc i = \mc r_2$ must hold,
    but the Boolean path may have $g_\mc i = \{g_{\mc{ir_2}}\}$ (describing the guard $\mc i \neq \mc r_1 \land \mc i = \mc r_2$).
    Thus, we got the following.

    \begin{obs}\label{obs:path-correspondence-A}\
      \li
      \- For every register automaton $A$, every data-path in $A$ has exactly one corresponding Boolean path in $A_\bbB$.
      \- There exists a register automaton $A$ where some Boolean paths of $A_\bbB$ do not correspond to any data-path of $A$.
      \il
    \end{obs}

    A \emph{Boolean word} is a projection of a Boolean path into $2^{P_\bbB}$;
    note that it contains information about assignment actions.

    Similarly we define Boolean transducers.
    Given a $k$-register transducer $T = \RTtuple$,
    a \emph{Boolean transducer} $T_\bbB = \tpl{I_\bbB, O_\bbB, S, s_0, \tau_\bbB}$ is a standard register-less transducer where:
    $I_\bbB = I \cup G_\mc i$, $G_\mc i = \{g_\mc{ir_1}, ..., g_{\mc{ir}_k}\}$,
    $O_\bbB = O \cup Asgn \cup O_k$, $Asgn = \{a_{\mc r_1}, ..., a_{\mc r_k}\}$, and
    $O_k$ has enough Boolean signals to encode the numbers $[k]$.
    The transition function $\tau_\bbB: S \x 2^{I_\bbB}  \to  S \x 2^{O_\bbB}$
    contains $(s, l\cup g_\mc i, o \cup o_k \cup a, s')$ iff $(s, l, \bar b_\mc i, o, \tilde o_k, \bar a, s') \in \tau$
    where $s,s' \in S$, $l \in 2^I$,
    $a \in 2^{\Asgn}$ encodes $\bar a \in \bbB^k$,
    $g_\mc i \in 2^{G_\mc i}$ encodes $\bar b_\mc i \in \bbB^k$,
    and $o_k \in 2^{O_k}$ encodes $\tilde o_k \in [k]$.
    A \emph{Boolean path} is an infinite sequence
    $s_0 \lblTo{l_0 \cup g_\mc i{}_0, o_0 \cup o_k{}_0 \cup a_0} s_1 \lblTo{l_1\cup g_\mc i{}_1, o_1 \cup o_k{}_1 \cup a_1} ...$
    from $(S \x 2^{I_\bbB} \x 2^{O_\bbB})^\omega$ that satisfies $\tau_\bbB$.

    Because every register transducer can be viewed as a register automaton,
    a similar observation holds for the register transducers.

  \subsection{Verifier to Remove Inconsistent Guards ($V_k$ and $A_\bbB @ V_k$)}\label{sec:AV}

    We introduce the automaton called verifier that filters out the Boolean paths of $A_\bbB$ that do not correspond to any data-paths.

    \parbf{$\mathbf{V_k}$}
    Given $k \in \bbN$,
    the \emph{verifier} is a deterministic looping register-less automaton
    $V_k = \tpl{P_V, \Pi, \pi_0, \delta_V}$ where
    \li
    \- $\Pi$ is the set of all possible partitions of $\{\mc r_1, ..., \mc r_k\}$;
       the initial state $\pi_0 = \{\{\mc r_1, ..., \mc r_k\}\}$ contains the only partition.
       Later, we will a partition-state to track if the registers have the same value.

    \- $P_V = G_\mc i \cup G_\mc o \cup \Asgn$ where $G_\mc i = \{g_{\mc i \mc r_1}, ..., g_{\mc i \mc r_k}\}$, $G_\mc o = \{g_{\mc o\mc r_1}, ..., g_{\mc o \mc r_k}\}$, $\Asgn = \{a_{\mc r_1}, ..., a_{\mc r_k}\}$.

    \- $\delta_V: \Pi \x 2^{P_V}  \to  \Pi$ contains $\pi \lblTo{{g_\mc i} \cup {g_\mc o} \cup a} \pi'$ where:
       \li
       \- the guard-letter $g_\mc i\cup g_\mc o$ respects the current partition:
          \li
          \- for every $\mc r_m = \mc r_n$ of $\pi$ (i.e., belonging to the same partition):\\
             $g_{\mc i \mc r_m}\!\! \in {g_\mc i} \Iff g_{\mc i \mc r_n}\!\! \in {g_\mc i}$ and
             $g_{\mc o \mc r_m}\!\! \in {g_\mc o} \Iff g_{\mc o \mc r_n}\!\! \in {g_\mc o}$;

          \- for every $\mc r_m \neq \mc r_n$ of $\pi$ (i.e., belonging to different partitions):\\
             $g_{\mc i \mc r_m} \!\!\in {g_\mc i} \Impl g_{\mc i \mc r_n} \!\!\not\in {g_\mc i}$ and
             $g_{\mc o \mc r_m} \!\!\in {g_\mc o} \Impl g_{\mc o \mc r_n} \!\!\not\in {g_\mc o}$;
          \il

       \- the successor partition respects the assignment-letter $a$, formalized as follows.
          For every $m$, $n$ in $[k]$,
          let $e_{mn}$ denote that $\pi$ contains $\mc r_m = \mc r_n$,
          and $e'_{mn}$ is for $\pi'$.
          The value $e'_{mn}$ is uniquely defined:
          $$
            e'_{mn} =
            (a_{\mc r_m} \land a_{\mc r_n}) \lor
            (\neg a_{\mc r_m} \land a_{\mc r_n} \land g_{\mc i \mc r_m}) \lor
            (a_{\mc r_m} \land \neg a_{\mc r_n} \land g_{\mc i \mc r_n}) \lor
            (\neg a_{\mc r_m} \land \neg a_{\mc r_n} \land e_{mn}).
          $$
          This definition, together with the previous item,
          ensures that all $e'_{mn}$ together form a partition
          (e.g., it is impossible to get $e'_{1,2} \land e'_{2,3} \land \neg e'_{1,3}$).

       \il

    \- The acceptance condition (not shown in the tuple) defines every path (infinite by definition) to be accepting;
       hence, every word that has a path in the automaton is accepted.
    \il
    An example of a verifier is in Figure~\ref{fig:verifier-automaton}.

    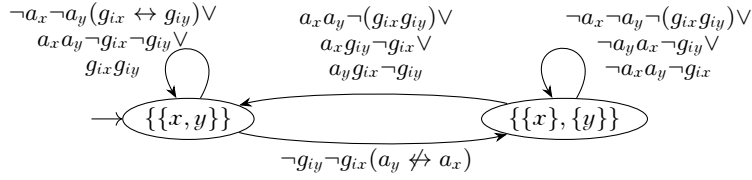
\begin{figure}[bt]
      \centering
      \begin{tikzpicture}
	\begin{pgfonlayer}{nodelayer}
		\node [style=text ellipse, initial] (0) at (-1, 1.25) {$\{\{\mc x,\mc y\}\}$};
		\node [style=text ellipse] (1) at (4, 1.25) {$\{ \{\mc x\}, \{\mc y\} \}$};
		\node [style=textual] (2) at (1.5, 2.25) {\specialcellC{$a_\mc x a_\mc y \neg (g_{\mc{ix}} g_{\mc{iy}}) \lor$\vspace{-0.5mm} \\$a_\mc x g_{\mc{iy}} \neg g_{\mc{ix}} \lor$\vspace{-0.5mm}\\$a_\mc y g_{\mc{ix}} \neg g_{\mc{iy}}$}};
		\node [style=textual, yshift=-0.5mm] (3) at (1.5, 0.75) {$\neg g_{\mc{iy}} \neg g_{\mc{ix}} (a_\mc y \not\iff a_\mc x)$};
		\node [style=textual, yshift=0.5mm] (4) at (-2, 2.25) {\specialcellC{$\neg a_\mc x \neg a_\mc y (g_{\mc{ix}} \iff g_{\mc{iy}}) \lor$\vspace{-0.5mm} \\ $a_\mc x a_\mc y \neg g_{\mc{ix}} \neg g_{\mc{iy}}\lor$\vspace{-0.5mm}\\$g_{\mc{ix}} g_{\mc{iy}}$}};
		\node [style=textual] (5) at (5.25, 2.25) {\specialcellC{$\neg a_\mc x \neg a_\mc y \neg (g_{\mc{ix}} g_{\mc{iy}}) \lor$\vspace{-0.5mm}   \\   $\neg a_\mc y a_\mc x \neg g_{\mc{iy}}\lor$\vspace{-0.5mm}\\    $\neg a_\mc x a_\mc y \neg g_{\mc{ix}}$}};
	\end{pgfonlayer}
	\begin{pgfonlayer}{edgelayer}
		\draw [style=arrow, in=120, out=60, loop] (1) to node[double]{} ();
		\draw [style=arrow, in=120, out=60, loop] (0) to ();
		\draw [style=arrow, bend right=15, looseness=0.50] (1) to (0);
		\draw [style=arrow, bend right=15, looseness=0.50] (0) to (1);
	\end{pgfonlayer}
\end{tikzpicture}
      \vspace{-2mm}
      \caption{A verifier automaton (a register-less deterministic looping automaton)
        for 2-register automata with $\mc R = \{\mc x, \mc y\}$.
        The edges have symbolic labels.
        Later,
        the left state $\{\{\mc x, \mc y\}\}$ will be used to denote that
        the registers $\mc x$ and $\mc y$ store the same value,
        while the right state $\{\{\mc x\}, \{\mc y\}\}$ will denote
        that they store different values.
        The automaton has similar restrictions for $\mc o$ (not shown).%
      }
      \label{fig:verifier-automaton}
    \end{figure}

    \parbf{$\mathbf{A_\bbB @ V_k}$}
    Given a verifier $V_k = \tpl{P^V, Q^V, q_0^V, \delta^V}$ and
    a register-less universal co-B\"uchi automaton $A_\bbB = \tpl{P^A, Q^A, q_0^A, \delta^A, F^A}$,
    let $A_\bbB @ V$ denote the universal co-B\"uchi automaton $\tpl{P, Q, q_0, \delta, F}$ where:
    \li
    \- $P = P^V \cup P^A$;
    \- $Q = Q^V \x Q^A$, $q_0 = (q_0^V, q_0^A)$;
    \- $\delta: Q \x 2^P \to 2^Q$ has $\big((q_V, q_A), p, (q'_V,q'_A)\big)$ iff
       $(q_V, p\cap 2^{P^V}, q'_V) \in \delta^V$ and $(q_A, p \cap 2^{P^A}, q'_A) \in \delta^A$; and
    \- $F = Q^V \x F^A$.
    \il
    Since $P^A = P' \cup G_\mc i \cup G_\mc o \cup \Asgn$ (where $P'$ are the Boolean signals of the register automaton $A$) and
    $P^V = G_\mc i \cup G_\mc o \cup \Asgn$,
    the automaton $A_\bbB @ V_k$ works on words from $(P' \cup G_\mc i \cup G_\mc o \cup \Asgn)^\omega$.
    The words of $A_\bbB @ V_k$ that do not fall out of $V_k$ are called \emph{consistent},
    otherwise \emph{inconsistent}.
    Notice that falling out of the verifier component favours accepting;
    $L(A_\bbB @ V_k) = \overline{L(V_k)} \cup L(A_\bbB)$,
    or, equivalently, $\overline{L(A_\bbB @ V_k)} = L(V_k) \cap \overline{L(A_\bbB)}$.
    Thus, the rejected words of $A_\bbB @ V_k$ are consistent and are rejected by $A_\bbB$.

    \begin{obs}\label{obs:path-correspondence-AV}
      For every universal co-B\"uchi $k$-register automaton $A$:
      \li
      \- every data-path of $A$ has exactly one corresponding Boolean path in $A_\bbB @ V_k$;
      \- every Boolean path of $A_\bbB @ V_k$ has either one or infinitely many corresponding data-paths in $A$.
      \il
    \end{obs}

    \begin{proof}[Proof]
    The first item follows from the definition of a data-pata.
    Consider the second item.
    Consider a Boolean path of $A_\bbB @ V_k$
    $$
      (q_0, \Pi_0) \lblTo{l_0\cup g_\mc i{}_0 \cup g_\mc o{}_0 \cup a_0} (q_1, \Pi_1) \lblTo{l_1 \cup g_\mc i{}_1 \cup g_\mc o{}_1 \cup a_1}...
    $$
    (where $q_j$ is a state of $A_\bbB$,
     $\Pi_j$ is a state of $V_k$,
     $l_j \in 2^P$,
     $g_\mc i{}_j \in 2^{G_\mc i}$,
     $g_\mc o{}_j \in 2^{G_\mc o}$, and
     $a_j \in 2^{\Asgn}$,
     for every $j \in \bbN_0$).
    We construct a corresponding data-path of $A$
    $$
      (q_0, \bar{\mc d}_0) \lblTo{(l_0, \mc i_0, \mc o_0, \bar a_0)} (q_1, \bar{\mc d}_1) \lblTo{(l_1, \mc i_1, \mc o_1, \bar a_1)}...:
    $$
    \li
    \- $\bar{\mc d}_0 = \mc d_0^k$;
    \- $\bar a_j \in \bbB^k$ encodes $a_j \in 2^{\Asgn}$,
    \- $\bar{\mc d}_{j+1}$ is uniquely defined by $\bar{\mc d}_j$, $\mc i_j$, and $\bar a_j$; and
    \- $\mc i_j$ and $\mc o_j$ are arbitrary such that
       $(\bar{\mc d}_j, \mc i_j, \mc o_j)$ satisfies the guards encoded by ${g_\mc i}_j$ and ${g_\mc o}_j$.
       Such values exist, because $\Pi_j$ and $g_\mc i{}_j$ and $g_\mc o{}_j$ are non-contradictory.
       Note that there are $>\!1$ possible values for $\mc i_j$ (in fact, infinitely many)
       iff ${g_\mc i}_j$ encodes the guard $\bigwedge_{m \in [k]} \mc i \neq \mc r_m$ (i.e., $\false^k$);
       similarly for $\mc o_j$.
    \il
    \end{proof}

    The observation, together with the definition of acceptance by $V_k$, implies the following.

    \begin{cor}\label{obs:AV-non-emptiness}
      For every universal co-B\"uchi $k$-register automaton $A$:\\
        $~~~~~~A_\bbB @ V_k$ has a rejected Boolean word $\Iff$ $A$ has a rejected data-word.
    \end{cor}

    If we look at the dual automaton $\bar A$ (non-deterministic B\"uchi) and the dual $\overline{A_\bbB @ V_k}$,
    then the corollary states that non-emptiness of non-deterministic B\"uchi register automata is decidable.
    This result was earlier established in~\cite[Thm.1]{Kaminski} using cutoffs (we discussed cutoffs on page~\pageref{page:cutoff}).
    Our verifier uses a similar insight, but it is handy in the context of synthesis.

    \ak{restore}

    \ak{restore}

  \subsection{Focusing on Transducer Data-Words ($\Tall$ and $A \ox \Tall$)}

    In the end, we will have a register-less automaton $H$,
    from which we will a Boolean associate of a register transducer.
    In the Boolean associate, the assignment actions are modelled as Boolean outputs.
    Therefore, the automaton $H$ should have Boolean signals expressing the assignment actions of the Boolean transducer.
    The automaton $\Tall$ fulfills this purpose:
    it adds $k_T$ fresh registers to $A$ that will be controlled by transducers via fresh Boolean signals.

    \parbf{$\mathbf{\Tall}$}
    Let $k_T \in \bbN$ and let $\Asgn^T = \{a_{\mc r^T_1}, ..., a_{\mc r^T_{k_T}}\}$ be fresh Boolean signals.
    $\Tall$ is a deterministic co-B\"uchi $k_T$-register automaton $\RAtuple$
    with $P = I \cup O \cup \Asgn^T$, $\mc P = \{\mc i, \mc o\}$, $Q = \{q_0, \lightning\}$, $F = \{\lightning\}$.
    The transition function
    $$
      Q \x 2^{I \cup O \cup \Asgn^T} \x \bbB_\mc i^{k_T} \x \bbB_\mc o^{k_T}  \to  Q \x \bbB^{k_T}
    $$
    \li
    \- for every $\bar g_\mc i \in \bbB_\mc i^{k_T}$,
       $\bar g_\mc o \in \{\bar g \in \bbB^{k_T} \| \exists j. \bar g[j] = \true\}$,
       and $a \in 2^{\Asgn^T}$,
       contains $(q_0, \bar a)$ where $\bar a[j]=\true$ iff $a_{\mc r^T_j} \in a$ for every $j \in [k_T]$;
    \- when $\bar g_\mc o$ does not satisfy the above condition, it transits from $q_0$ to $\lightning$;
    \- it self-loops in $\lightning$ without storing for every letter.
    \il
    In words:
    $\Tall$ ensures that the value of data-output $\mc o$ comes from a register and
    the assignment actions are \emph{synced} with the Boolean signals $\Asgn^T$.

    \begin{obs}
      Let $k_T \in \bbN$, then: for every $w \in (2^{I \cup O \cup \Asgn^T} \x \mc D^2)^\omega$:
      $$
        w \models \Tall ~\Iff~ \exists T\!\!: w \in L(T),
      $$
      where $T$ is a $k_T$-register transducer (possibly, $|S| = \infty$)
      whose output is extended with $\Asgn^T$ signals that are synced with $T$'s assignment actions.
    \end{obs}
    \noindent
    In the observation, $T$ might need infinitely many states,
    because an accepting path of $\Tall$ on $w$ might exhibit ``irregular'' storing behaviour,
    which cannot be expressed by a finite-state transducer (recall that transducers are deterministic).
    That is a minor technical detail though.

    \parbf{$\mathbf{A \ox \Tall}$}
    The product $A \ox \Tall$ of
    a universal co-B\"uchi register automaton $A=\RAtupleA$ and $\Tall=\TALLtuple$,
    where $P^T = P^A \cup \Asgn^T$,
    is a universal co-B\"uchi $(k_A+k_T)$-register automaton \\ $\RAtuple$,
    where $P = P^T$, $\mc R = \mc R^A\cupdot \mc R^T$, $Q = Q^A \xx Q^T$, $q_0 = (q_0^A, q_0^T)$, $F = F^A \xx Q^T \cup Q^A \xx F^T$,
    and the transition function
    $$
      \delta: Q \x 2^{I \cup O \cup \Asgn^T} \x \bbB_\mc i^{k_A + k_T} \x \bbB_\mc o^{k_A + k_T}  \to  2^{Q \x \bbB^{k_A}} \!\!\x\! \bbB^{k_T}
    $$
    respects both $\delta^A$ and $\delta^T$.

    \begin{obs}
      For every $k_T\in \bbN$, universal co-B\"uchi $k_A$-register automaton $A$, and $w \in (2^{P^A \cup \Asgn^T} \x \mc D^2)^\omega$:
      $$
        w \models A \ox \Tall
        ~\Iff~
        w \models \Tall \textit{ and } w|_{2^{P^A}} \models A,
      $$
      where $w|_{2^{P^A}}$ is a projection of $w$ into $2^{P^A}$.
    \end{obs}

    \ak{note: env cannot force falling out of $\Tall$}


  \subsection{Synthesis-tailored Verifier ($AT_\bbB @ W$)} \label{sec:ATW}

    For brevity, let $AT$ denote $A\ox \Tall$, and let $AT_\bbB$ be its Boolean associate.

    The automaton $AT_\bbB @ W$ that will be introduced in this section closely resembles $AT_\bbB @ V_k$ and $A_\bbB @ V_k$,
    but it is better suited for synthesis.

    Recall from Section~\ref{sec:associates} that every $T_\bbB$ generates words from $(2^{I \cup G^T_\mc i} \x 2^{O \cup \Asgn^T \cup O_{k_T}})^\omega$,
    where $\Asgn^T = \{a_{\mc r^T_1}, ..., a_{\mc r^T_{k_T}}\}$, $G^T_\mc i = \{g_{\mc i \mc r^T_1}, ..., g_{\mc i \mc r^T_{k_T}}\}$,
    and $O_{k_T}$ has enough Boolean signals to encode the numbers $[k_T]$.
    For synthesis we want our target specification automaton to have the same alphabet.
    The automaton $AT_\bbB @ V_k$ uses $\mc o$-guards instead of signals $O_k$,
    hence we introduce the automaton $AT_\bbB @ W$
    (we do not introduce $W$ separately).

    Suppose we have $\ATV = \tpl{P, Q, q_0, \delta, F}$
    with $P = I \cup O \cup G_\mc i^T \cup G_\mc i^A \cup G_\mc o^T \cup G_\mc o^A \cup \Asgn^T \cup \Asgn^A$
    and $\delta: Q \x 2^P \to 2^Q$.
    The automaton $\ATW = \tpl{P', Q, q_0, \delta', F}$ has the same states,
    but $P' = (P \setminus (G^T_\mc o \cup G^A_\mc o)) \cup O_{k_T}$ and
    the transition function $\delta'$ is derived from $\delta$ as follows.
    For every $(\pi, q) \lblTo{(i, o, g_\mc i, g_\mc o, a)} (\pi', q')$ of $\delta$
    (where $\pi$ and $\pi'$ are partitions of $\mc R^A \cup \mc R^T$, $q$ and $q'$ are states of $AT_\bbB$,
     $i \in 2^I$, $o \in 2^O$, $g_\mc i \in 2^{G^A_\mc i\cup G^T_\mc i}$, $g_\mc o \in 2^{G^A_\mc o \cup G^T_\mc o}$, $a \in 2^{\Asgn^A\cup\Asgn^T}$):
    \li
    \- let $J = \{j_1, ..., j_l\} \subset \bbN$ be such that
       $g_\mc o$ contains $\mc o = \mc r^T_j$ for every $j \in J$;
    \- for every $j \in J$,
       add to $\delta'$ the transition $(\pi, q)\lblTo{(i,o,g_\mc i,\tilde j,a)}(\pi', q')$,
       where $\tilde j \in 2^{O_{k_T}}$ encodes the number $j \in [k_T]$.
    \- Note that if $J$ is empty ($g_\mc o$ requires that $\bigwedge_{t \in [k_T]} \mc o \neq \mc r^T_t$),
       then we do not add transitions to $\delta'$,
       because no transducer can produce such a value for $\mc o$.
    \il


    Notice that $\ATW$, just like $\ATV$, accepts inconsistent words (those fall out of the original $V_k$).
    Inconsistency in those words can come from signals $G_\mc i^A \cup G_\mc i^T$.
    Later, these Boolean signals will either be hidden ($G_\mc i^A$) or under environment control ($G_\mc i^T$),
    which means that a transducer will not be able to sabotage the specification by producing inconsistent words.

    The following observation resembles Observation~\ref{obs:path-correspondence-AV},
    but focuses on $k_T$-register transducers.

    \begin{obs}\label{obs:paths-AT-ATW}
      For every universal co-B\"uchi $k_A$-register automaton $A$, $k_T\in\bbN$:
      \li
      \- every data-path of $\AT$ has exactly one corresponding Boolean path in $\ATW$;
      \- every Boolean path of $\ATW$ has either one or infinitely many corresponding data-paths in $\AT$.
      \il
    \end{obs}





  \subsection{Synthesis Using Automaton $\hideATW$}

    We cannot use $\ATW$ for synthesis,
    because it uses Boolean signals that are not visible to transducers (underlined):
    $I \cup O \cup \underline{G_\mc i^A} \cup G_\mc i^T \cup \underline{G_\mc o^A} \cup O_{k_T} \cup \underline{\Asgn^A} \cup \Asgn^T$.
    Let us show that the simple hiding operation resolves the issue.

    Given $\ATW = \tpl{P, Q, q_0, \delta, F}$ with $P = I \cup O \cup G_\mc i^A \cup G_\mc i^T \cup G_\mc o^A \cup O_{k_T} \cup \Asgn^A \cup \Asgn^T$,
    the automaton $\hideATW$ is a universal co-B\"uchi automaton $\tpl{P', Q, q_0, \delta', F}$
    with $P' = I \cup O \cup G_\mc i^T\!\cup O_{k_T}\!\cup \Asgn^T$ and
    $$
      \delta': Q \x 2^I \x 2^O \x 2^{G_\mc i^T} \x 2^{O_{k_T}} \x 2^{\Asgn^T} \to  2^Q
    $$
    consists of transitions $q \lblTo{(i,o,g_\mc i^T,j,a^T)} Q'$ that satisfy the following:
    the destination set $Q' \subseteq Q$ contains all successor states
    of every transition of $\ATW$ starting in $q$ and having the same common labels:
    $$
      Q' ~= \!\!\!\!\!\!\!\!\!\!\!\!\!\!\!\!\!\!\!\!\!\!\!\!\!
            \bigcup_{g_\mc i^A \in 2^{G_\mc i^A},g_\mc o^A \in 2^{G_\mc o^A},a^A\in 2^{\Asgn^A}}
           \!\!\!\!\!\!\!\!\!\!\!\!\!
           \!\!\!\!\!\!\!\!\!\!\!\!\!
            \delta(q,i,o,g_\mc i^A, g_\mc i^T,g_\mc o^A,j,a^T,a^A).
    $$

    \begin{obs}\label{obs:paths-ATW-hideATW} \
      For every universal co-B\"uchi register automaton $A$, $k_T \in \bbN$:
      \li
      \- every path of $\ATW$ corresponds to exactly one path of $\hideATW$;
      \- every path of $\hideATW$ corresponds to at least one path of $\ATW$.
      \il
    \end{obs}

    \begin{proof}
      The first item follows from the definition of $\hideATW$.

      Consider the second item.
      Fix a path $p = q_1 \lblToS{\sigma_1} q_2 \lblToS{\sigma_2} ...$ of $\hideATW$.
      By definition, for every transition $q_j \lblToS{\sigma_j} q_{j+1}$ of $\hideATW$,
      there must be some transition $q_j \lblToS{\sigma'_j} q_{j+1}$ of $\ATW$,
      where $\sigma'_j$ and $\sigma_j$ agree on the values of shared signals.
      Hence, in order to get the desired path of $\ATW$, we do the following:
      for every $j$, \emph{arbitrary} choose
      $\sigma'_j \in 2^{I \cup O \cup {G_\mc i^A} \cup G_\mc i^T \cup {G_\mc o^A} \cup O_{k_T} \cup {\Asgn^A} \cup \Asgn^T}$
      that satisfies $\delta_{\ATW}$ and agrees with $\sigma_j \in 2^{I \cup O \cup G_\mc i^T \cup O_{k_T} \cup \Asgn^T}$
      on the values of shared signals.
    \end{proof}

    \begin{lemma}\label{lem:rej-words}
      For every $k_T$-register transducer $T$ and universal co-B\"uchi $k_A$-register automaton $A$:
      $$
        \big(\exists w \in L(T): w \not\models A\big)  ~\Iff~  \big(\exists w_\bbB \in L(T_\bbB): w_\bbB \not\models \hideATW\big).
      $$
    \end{lemma}

    \begin{proof}
      Both directions follow from the definitions and Observations~\ref{obs:paths-AT-ATW}~and~\ref{obs:paths-ATW-hideATW}.

      Consider direction $\Leftarrow$.
      The word
      $w_\bbB\in (2^{I \cup O \cup G_\mc i^T \cup O_{k_T} \cup \Asgn^T})^\omega$
      induces a path
      $\pi_{t_b} \in (S\x 2^{I\cup O \cup G_\mc i^T \cup O_{k_T} \cup \Asgn^T})^\omega$ on $T_\bbB$ and
      a rejected path
      $\pi_h \in (Q_h \x 2^{I\cup O \cup G_\mc i^T \cup O_{k_T} \cup \Asgn^T})^\omega$
      on $\hideATW$.
      By Observation~\ref{obs:paths-ATW-hideATW}, $\pi_h$ corresponds to at least one path
      $\pi_{atw} \in (Q_h \x 2^{I \cup O \x G_\mc i^A \cup G_\mc i^T \cup G_\mc o^A \cup O_{k_T} \cup \Asgn^A \cup \Asgn^T})^\omega$
      of $\ATW$.
      By Observation~\ref{obs:paths-AT-ATW}, $\pi_{atw}$ corresponds to at least one data-path
      $\pi_{at} \in (Q_{at} \x 2^{I \cup O \cup \Asgn^T} \x \mc D^2)^\omega$ of $\AT$,
      which is rejected by $A$,
      because $\pi_h$ is rejected by $A_\bbB$.
      Thus, we get $w \in (2^{I \cup O} \x \mc D^2)^\omega$ from $\pi_{at}$ by projecting,
      which completes the direction.
      Notice that a data-path
      $\pi_t \in (S \x 2^{I \cup O \cup \Asgn^T} \x \mc D^2)^\omega$ of $T$
      induced by $w$
      corresponds to the Boolean path $\pi_{t_b}$ of $T_\bbB$
      induced by $w_\bbB$,
      despite the particular choices of $\pi_{atw}$ and $\pi_{at}$.

      The other direction is similar.
    \end{proof}

    The lemma implies a solution to the bounded synthesis problem.
    \begin{theorem}\label{thm:synthesis-reduction}
      For every universal co-B\"uchi register automaton $A$ and $k_T \in \bbN$:
      $$
        \big(\exists T: T \models A\big)  ~\Iff~  \big(\exists T_\bbB: T_\bbB \models \hideATW\big),
      $$
      where $T$ is a $k_T$-register transducer.
    \end{theorem}

    The right side of the theorem (the standard Boolean synthesis problem)
    holds
    iff
    it holds for finite-state transducers (e.g., see \cite{DBLP:conf/popl/PnueliR89}).
    Hence we get:
    
    \begin{cor}
      A given instance of the bounded synthesis problem is realizable
      $\Iff$
      it is realizable by a finite-state ($|S| < \infty$) register transducer.\ak{how many states?}
    \end{cor}

    Let us consider the complexity of our approach.
    The automaton $\hideATW$ has $|Q_A| \cdot |\Pi|$ states,
    where $Q_A$ is the number of states in $A$ and $|\Pi|$ is the number of partitions of the set $\{1, ..., k\}$ where $k = k_T + k_A$.
    The latter is a Bell number~\cite{bell-number} and is
    less than $(\frac{0.792k}{\ln(k+1)})^k$~\cite[Thm~2.1]{berend2010improved}.
    Hence the number of states in $\hideATW$ is less than $|Q_A|\cdot (\frac{0.792k}{\ln(k+1)})^k$,
    and the complexity of our approach is in $synth(|Q_A|\cdot (\frac{0.792k}{\ln(k+1)})^k)$,
    where $synth(m) = 2^{c\cdot m}$ is the complexity of synthesis
    from a universal co-B\"uchi automaton with $m$ states~\cite[Thm.2]{DBLP:conf/popl/PnueliR89}
    ($c$ is a constant).
    This is an upper bound, the lower bound is open, thus we get:

    \begin{cor}
      The bounded synthesis problem can be solved in $2^{c\cdot |Q_A|\cdot (\frac{0.792k}{\ln(k+1)})^k}$ time,
      where $k = k_A + k_T$, $|Q_A|$ and $k_A$ is the number of states and registers in a given universal automaton,
      and $c$ is a constant.
    \end{cor}

    Finally, Figure~\ref{fig:languages-diagram} depicts the relation between the languages of utilized automata.
    It shows that the approach makes use of determinizable subset of $\AT$.

    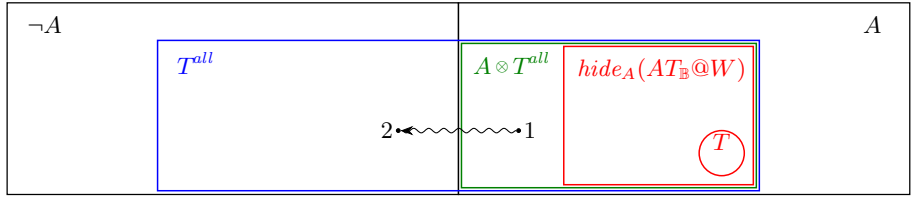
\begin{figure}[tb]
      \centering
      \begin{tikzpicture}

    \def\setA{(0,-0.05) rectangle (6,2.5)}
    \def\setNotA{(-6,-0.05) rectangle (0,2.5)}
    \def\setTall{(-4,0) rectangle (4,2)}
    \def\setAT{(0.04,0.04) rectangle (3.96,1.96)}
    \def\setTreat{(1.4,0.08) rectangle (3.92, 1.92) }
    \def\setT{(3.5,0.5) circle (0.3)}
    
    \begin{scope}[fill opacity=0.3]


        \draw[line width=0.17mm, black] \setA;
        \draw[line width=0.17mm, black] \setNotA;
        \draw[line width=0.20mm, blue] \setTall;
        \draw[line width=0.20mm, green] \setAT;
        \draw[line width=0.20mm, red] \setTreat;
        \draw[line width=0.20mm, red] \setT;

    \end{scope}

    \node [style=textual] (tA) at (5.5, 2.2) {$A$};
    \node [style=textual] (tNotA) at (-5.5, 2.2) {$\neg A$};
    \node [style=textual] (tAT) at (0.7, 1.7) {\color{Green}{$A \ox \Tall$}};
    \node [style=textual] (tTall) at (-3.5, 1.7) {\color{Blue}{$\Tall$}};

    \node [style=textual] (tTreat) at (2.7, 1.6) {\color{Red}{$\hideATW$}};

    \node [style=textual] (tT) at (3.5, 0.65) {\color{Red}{$T$}};

    \node [style=dot] (1) at (0.8, 0.8) {};
    \node [style=textual] (t1) at (0.95, 0.8) {$1$};

    \node [style=dot] (2) at (-0.8, 0.8) {};
    \node [style=textual] (t2) at (-0.95, 0.8) {$2$};

    \draw [style=arrow,decorate,decoration={snake,amplitude=.3mm,segment length=2mm}] (1) to (2);

\end{tikzpicture}
      \vspace{-2mm}
      \caption{%
        Inclusion between languages.
        The automaton $\hideATW$ is Boolean, but here it is viewed as a register automaton.
        Also, the alphabet of $A$ is extended with $\Asgn^T$ to coincide with that of $\AT$ and $\hideATW$.
        Figure~\ref{fig:proof-automaton-4-5} justifies the existence of point 1,
        which explains why $\hideATW$ can be a \emph{strict} subset of $\AT$.
        The snake line indicates
        ``for every $T$: if it has point 1, then it also has point 2''
        (by Lemma~\ref{lem:rej-words}).
        Thus, if $T \models A$ for some $k_T$-register transducer,
        then it must be located inside $\hideATW$.%
      }
      \label{fig:languages-diagram}
    \end{figure}
    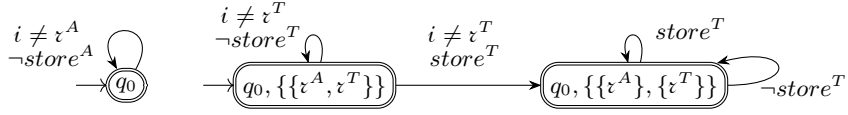
\begin{figure}[tb]
      \centering
      \begin{tikzpicture}
	\begin{pgfonlayer}{nodelayer}
		\node [double, initial, style=text rounded] (0) at (-2, 1.25) {$q_0$};
		\node [style=textual, yshift=0mm] (1) at (-3, 1.75) {\specialcellC{$\mc i \neq \mc r^A$ \vspace{-0.8mm} \\ $\neg store^A$}};
		\node [style=text rounded, initial, double] (2) at (0.5, 1.25) {$q_0,\{ \{\mc r^A, \mc r^T\}\}$};
		\node [double, style=text rounded] (3) at (4.75, 1.25) {$q_0, \{\{\mc r^A\}, \{\mc r^T\}\}$};
		\node [yshift=0mm, style=textual] (4) at (-0.25, 2) {\specialcellC{$\mc i \neq \mc r^T$ \vspace{-0.8mm} \\ $\neg store^T$}};
		\node [style=textual, yshift=0mm] (5) at (2.5, 1.75) {\specialcellC{$\mc i \neq \mc r^T$ \vspace{-0.8mm} \\ $store^T$}};
		\node [yshift=0mm, style=textual] (6) at (5.5, 2) {$store^T$};
		\node [style=textual, yshift=0mm] (7) at (7, 1.25) {$\neg store^T$};
	\end{pgfonlayer}
	\begin{pgfonlayer}{edgelayer}
		\draw [style=arrow, in=120, out=60, loop] (0) to ();
		\draw [style=arrow, in=105, out=75, loop] (2) to ();
		\draw [style=arrow, in=105, out=75, loop] (3) to ();
		\draw [style=arrow] (2) to (3);
		\draw [style=arrow, in=15, out=0, loop] (3) to ();
	\end{pgfonlayer}
\end{tikzpicture}
      \vspace{-2mm}
      \caption{Universal co-B\"uchi register automata to show the existence of point 1 in Fig.\ref{fig:languages-diagram}.
        On the left is 1-register automaton $A$:
        it accepts the words where at some moment the signal $\mc i$ equals to $\mc d_0$ (and no restrictions on the values of $\mc o$).
        On the right is $\hideATW$ where $k_T=1$:
        when viewed as a register automaton,
        it accepts the words where the first value of $\mc i$ is $\mc d_0$ (plus some restrictions on $\mc o$).
        Hence, $L(\hideATW) \subsetneq L(A \ox \Tall)$.
        The labels related to $\mc o$ are omitted.}
      \label{fig:proof-automaton-4-5}
    \end{figure}

\section{Using Temporal Logic in our Synthesis Approach}\label{sec:logic}
  
  We proceed to the topic of synthesis of register transducers from a temporal logic.
  Section~\ref{sec:ltleq} defines a first-order linear temporal logic with equality,
  LTL(EQ)\footnote{The name LTL(EQ) is inspired by the names of logics in SMT-LIB~\cite{smtlib}.}
  and its variants $\exists$LTL(EQ) and $\forall$LTL(EQ),
  known as IPTL in \cite{Wolper-data-indep} and VLTL in \cite{grumberg2012model}.
  Then Section~\ref{sec:gue-atm} defines register-guessing automata
  that can express $\exists$LTL(EQ) formulas.
  The sound and complete conversion of $\exists$LTL(EQ) into register-guessing automata is described in Section~\ref{sec:ltleq-to-gue-atm}.
  Then Section~\ref{sec:ltleq-to-reg-atm} describes a sound but incomplete conversion of register-guessing automata into register automata,
  which implies the sound but incomplete conversion of $\exists$LTL(EQ) into register automata
  (no complete conversion can exist).
  The latter automata are consumed by our synthesizer.

  Unless explicitly stated, all automata are non-deterministic B\"uchi.

  \subsection{LTL(EQ) (also known as IPTL~\cite{Wolper-data-indep} and VLTL~\cite{grumberg2012model})}\label{sec:ltleq}
  
    Let $\mc X$ be a set of data-variables and $P$ be a set of Boolean propositions.
    An \emph{LTL(EQ) (prenex-quantified) formula} $\Phi$ is of the form (for every $k \in \bbN$):
    \begin{align*}
      \Phi  &~=~  \forall \mc x_1... \mc x_k. cond. \varphi \| \exists \mc x_1...\mc x_k. cond. \varphi \\
       cond &~=~  \true \| \mc x \neq \mc x \| cond \land cond \\
    \varphi &~=~  \true \|
                p \|
                \mc i = \mc x \| \mc o = \mc x \|
                \neg \varphi \| \varphi \land \varphi \| 
                \varphi \U \varphi \| \X \varphi
    \end{align*}
    where $\mc x_1, ..., \mc x_k, \mc x \in \mc X$,
    $p \in P$,
    $\mc i$ and $\mc o$ are two data-propositions,
    and all the data-variables appearing in $\varphi$ are quantified.
    As usual,
    define $\G \varphi$ to be $\neg\F\varphi$,
    $\F\varphi = \true \U \varphi$,
    $\varphi_1 \lor \varphi_2$ is $\neg (\neg\varphi_1 \land \neg\varphi_2)$,
    $\varphi_1 \impl \varphi_2$ is $\neg\varphi_1 \lor \varphi_2$, and
    $\false$ is $\neg \true$.

    Given $w = w_1 w_2 ... \in (2^P \x \mc D^{\{\mc i, \mc o\}})^\omega$,
    define the satisfaction $w \models \Phi$:
    \li
    \- $w \models \forall \mc x_1...\mc x_k. cond. \varphi$
       iff
       for all $\mc d_1,...,\mc d_k \in \mc D$
       either $cond[\mc x_1 \from \mc d_1,...,\mc x_k \from \mc d_k]$ does not hold or
       $w \models \varphi[\mc x_1\from \mc d_1, ..., \mc x_k \from \mc d_k]$;

    \- $w \models \exists \mc x_1...\mc x_k. cond. \varphi$
       iff
       there exists $\mc d_1,...,\mc d_k \in \mc D$
       such that $cond[\mc x_1 \from \mc d_1,...,\mc x_k \from \mc d_k]$ holds and
       $w \models \varphi[\mc x_1\from \mc d_1, ..., \mc x_k \from \mc d_k]$;

    \- let $\phi$ have the same grammar as $\varphi$ except that instead of data-variables it has data-values;
       then
    \- $w \models \true$;
    \- $w \not\models \phi$ iff $\neg(w \models \phi)$;
    \- $w \models \neg \phi$ iff $\neg(w \models\phi)$;
    \- $w \models p$ iff $p \in w_1$;
    \- $w \models \phi_1 \land \phi_2$ iff $w \models \phi_1$ and $w \models \phi_2$;
    \- for every $\mc d \in \mc D$,
       $w \models \mc i = \mc d$ iff in $w_1$ the data-proposition $\mc i$ has the value $\mc d$;
       similarly for $\mc o$;
    \- for $i \in \bbN$, let $w_{[i:]}$ denote $w$'s suffix $w_i w_{i+1} ...$; then
    \- $w \models \X \phi$ iff $w_{[2:]} \models \phi$; and
    \- $w \models \phi_1 \U \phi_2$ iff $\exists i \in \bbN: \big((w_{[i:]} \models \phi_2) \land (\forall j<i: w_{[j:]} \models \phi_1)\big)$.
    \il

    Let $\exists$LTL(EQ) denote LTL(EQ) where formulas have existential quantifiers only,
    and use $\forall$LTL(EQ) for universally quantified LTL(EQ) formulas.

  \subsection{Register Automata with Guessing but Without Storing}\label{sec:gue-atm}
    
    In this section we define a variation of register automata that
    have a non-deterministically chosen initial register values
    that cannot be rewritten afterwards.
    Such automata are a restricted version of variable automata~\cite{grumberg2010variable}.

    A \emph{$k$-register-guessing automaton} is a tuple $A = \VAtuple$
    (notice: no initial register value $\mc d_0$ and a new element $E$)
    with transition function $\delta$ of the form
    $
        Q \x
        2^{P} \x 
        \bbB_\mc i^k \x
        \bbB_\mc o^k
        \to
        2^{Q}
    $
    (notice: no assignment component on the right),
    where $E \subseteq \mc R \x \mc R$ is an \emph{inequality set}\footnote{%
      We can get away without using $E$ (by encoding it into $\delta$),
      but it proved to be convenient in Section~\ref{sec:ltleq-to-reg-atm}.%
    },
    while all other components are like for register automata. 
    A path is defined similarly to a path of a register automaton,
    except that
    \li
    \- an initial configuration $(q_0, \mc{\bar d}_0) \in \{q_0\} \x \mc D^k$ of the path is arbitrary
       provided that $\mc{\bar d}_0$ satisfies the inequality set:
       $\forall (\mc r_i, \mc r_j) \in E: \mc{\bar d}_0[i] \neq \mc{\bar d}_0[j]$; and
    \- the automaton never stores to the registers.
    \il
    An accepting word is defined as for register automata.

  \subsection{Converting $\exists$LTL(EQ) into Register-Guessing Automata}\label{sec:ltleq-to-gue-atm}

    This section describes the conversion of $\exists$LTL(EQ) formulas into register-guessing automata with the same language.
    The fact that a conversion is possible was noted in~\cite[Sec.4]{10.1007/978-3-319-57288-8_1},
    however they did not describe the conversion itself.

    Consider an $\exists$LTL(EQ) formula $\Phi = \exists \mc x_1...\mc x_k. cond. \varphi(\mc i, \mc o, \mc x_1, ..., \mc x_k)$.
    \newcommand\BoolArgs{g_{\mc i\mc r_1} ,..., g_{\mc i\mc r_k}, g_{\mc o\mc r_1} ,..., g_{\mc o\mc r_k}}
    We will use the notions of $w_\bbB$ and $\varphi_\bbB$ defined below.
    \li
    \-[($w_\bbB$)] Given a word $w \in (2^P \x \mc D^2)^\omega$ and $\mc x_1, ..., \mc x_k \in \mc D$,
       let $w_\bbB \in (2^P \x \bbB_\mc i^k \x \bbB_\mc o^k)^\omega$ be the word derived from $w$
       by replacing every value of $\mc i$ and $\mc o$ in $w$
       by the vectors of Boolean values, $(\mc i = \mc x_1,...,\mc i = \mc x_k)$ and $(\mc o = \mc x_1,...,\mc o = \mc x_k)$.

    \-[($\varphi_\bbB$)] In $\varphi(\mc i, \mc o, \mc x_1, ..., \mc x_k)$,
       replace every expression $\mc i = \mc x_i$ with a new literal $g_{\mc i\mc r_i}$
       and every expression $\mc o = \mc x_i$ with $g_{\mc o\mc r_i}$.
       This introduces $2k$ new Boolean propositions,
       let $P_\bbB = P \cup \{g_{\mc i\mc r_1} ,..., g_{\mc i\mc r_k}\} \cup \{g_{\mc o\mc r_1} ,..., g_{\mc o\mc r_k}\}$.
       Let $\varphi_\bbB(\BoolArgs)$ be the resulting LTL formula over Boolean propositions $P_\bbB$.
    \il
    \smallskip
    \noindent To convert a formula $\exists \mc x_1...\mc x_k.cond.\varphi$ into a $k$-register-guessing automaton $A$ do the following
    ({\bf conversion-1}).
    \li
    \- Convert $\varphi_\bbB$ into an NBW automaton $A_\bbB=\tpl{P_\bbB, Q, q_0, \delta_\bbB, F}$ using standard approaches.
       Thus, for every $w_\bbB \in 2^{P_\bbB}$:
       $w_\bbB \models A_\bbB$ iff $w_\bbB \models \varphi_\bbB$.

    \- Treat $A_\bbB$ as a $k$-register-guessing automaton $A = \VAtuple$,
       where $E$ is derived from $cond$.
    \il

    \noindent
    For example, the automaton in Figure~\ref{fig:reg-guessing-automaton} expresses the formula
    \begin{equation*}
      \neg
      \forall \mc x_1 \neq \mc x_2.
        \G\left[
               \begin{aligned}
                 & \mc i = \mc x_1 \land \X \mc i = \mc x_2 \impl \XX \neg e \\
                 & \mc i = \mc x_1 \land \X \mc i = \mc x_1 \impl \XX (e \land \mc o = x_1)
               \end{aligned}
          \right]
    \end{equation*}
    that says:
    compare the data-input $\mc i$ at two consecutive points and then
    (i) whenever they are equal, raise $e$ and output the data,
   (ii) otherwise, lower $e$.

    \begin{figure}[bt]
      \centering
      \begin{tikzpicture}
	\begin{pgfonlayer}{nodelayer}
		\node [style=wn, initial] (0) at (-1, 1) {$q_0$};
		\node [style=wn] (1) at (1, 2) {$q_1$};
		\node [style=textual] (2) at (-1, 2) {$\true$};
		\node [style=textual, rotate=25, yshift=-1mm, xshift=1mm] (3) at (-0.25, 1.75) {$\mc i = \mc r_1$};
		\node [style=wn] (4) at (3, 2) {$q_3$};
		\node [double, style=wn] (5) at (5, 1) {$q_5$};
		\node [style=wn] (6) at (1, -0) {$q_2$};
		\node [style=wn] (7) at (3, -0) {$q_4$};
		\node [style=textual] (8) at (5, 2) {$\true$};
		\node [xshift=1mm, yshift=-1mm, rotate=-25, style=textual] (9) at (-0.25, 0.5) {$\mc i = \mc r_1$};
		\node [style=textual] (10) at (2, 2.25) {$\mc i = \mc r_2$};
		\node [style=textual] (11) at (2, -0.25) {$\mc i = \mc r_1$};
		\node [xshift=1mm, yshift=-1mm, rotate=25, style=textual] (12) at (3.75, 0.75) {$\neg e$};
		\node [style=textual, rotate=-25, yshift=-1mm, xshift=1mm] (13) at (3.75, 2) {$e$};
		\node [style=textual, rotate=29] (14) at (4.5, -0) {$\mc o \neq \mc r_1$};
	\end{pgfonlayer}
	\begin{pgfonlayer}{edgelayer}
		\draw [style=arrow] (0) to (1);
		\draw [style=arrow, in=120, out=60, loop] (0) to ();
		\draw [style=arrow, in=120, out=60, loop] (5) to ();
		\draw [style=arrow] (0) to (6);
		\draw [style=arrow] (6) to (7);
		\draw [style=arrow, bend right=45, looseness=0.75] (7) to (5);
		\draw [style=arrow] (1) to (4);
		\draw [style=arrow] (4) to (5);
		\draw [style=arrow] (7) to (5);
	\end{pgfonlayer}
\end{tikzpicture}
      \caption{A 2-register-guessing automaton: $P = \{e\}$, $\mc R = \{\mc r_1, \mc r_2\}$, $E=\{(\mc r_1, \mc r_2)\}$.
        The edges have symbolic labels,
        e.g., the edge labeled with $\mc i = \mc r_1$ encodes 16 edges,
        for different valuations of $e$, $\mc i = \mc r_2$, $\mc o = \mc r_1$, and $\mc o = \mc r_2$.%
      }
      \label{fig:reg-guessing-automaton}
    \end{figure}
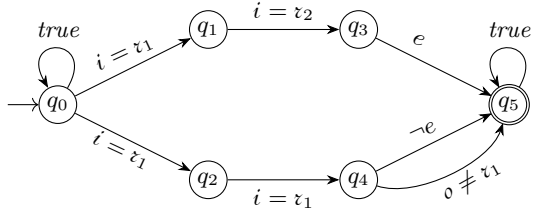

    \begin{obs}\label{obs:efoltl--gue-reg-atm}
      For every $w \in (2^P \x \mc D^2)^\omega$:
      $
        w \models A  ~\Iff~  w \models \exists \mc x_1...\mc x_k.cond.\varphi.
      $
    \end{obs}
    \ak{restore proof}


  \subsection{Converting $\exists$LTL(EQ) into Register Automata} \label{sec:ltleq-to-reg-atm}
    
    In this section,
    we describe a sound but incomplete conversion of register-guessing automata
    into standard register automata.
    Together with conversion-1 from the previous section,
    this gives the conversion of $\exists$LTL(EQ) formulas into register automata.
    Note that no complete conversion of $\exists$LTL(EQ) formulas into register automata exists:
    for example,
    the formula $\exists \mc x. \G (\mc i \neq \mc x)$ has no equivalent register automaton\ak{how to prove?:)},
    although there is an equivalent register-guessing automaton.

    In automata, we will use the definition of $\delta$ that is symbolic instead of explicit,
    hence the transition functions of $k$-register-guessing automata and of $k$-register automata are of the form
    $
      Q \x 2^P \x G  \to  2^Q \text{~~and~~} Q \x 2^P \x G  \to  2^{Q\x\bbB^k},
    $
    (previously we had $\bbB_\mc i^k \x \bbB_\mc o^k$ instead of $G$),
    where $g \in G$ has the form
    $g = \true \| g \land g \| \mc i \sim \mc r \| \mc o \sim \mc r$
    where $\sim$ denotes $=$ or $\neq$, and $\mc r \in \mc R$.
    Using the symbolic definition rather than the explicit one is crucial
    in making our conversion more applicable.
    \ak{restore the prev sentence}
    \ak{what does this \emph{really} mean?}

    Given a $k$-register-guessing automaton $A = \VAtuple$,
    construct the $k$-register automaton $A' = \RAtuplePrime$
    ({\bf conversion-2}):
    \li
    \- $Q' = Q \times \bbB^k$.
       The Boolean component encodes, for every $\mc r_i \in \mc R$,
       whether the register $\mc r_i$ is assigned a value or not (ignoring the initial values).
       The initial state $q'_0 = (q_0,\false,...,\false)$.
       We call a register $\mc r_i$ with $b_i = \false$ \emph{uninitialized}.

    \- $F' = \{(q, b_1,...,b_k) \in Q' \| q \in F\}$.

    \- For every state $(q,b_1,...,b_k) \in Q'$ and
       $A$-transition $q \lblTo{(l, g)} q'$ ($l \in 2^P$, $g \in G$):
       \li
       \- If $g = \true$,
          then add to $\delta'$ the transition
          $
            (q,b_1,...,b_k) \lblTo{(l, g, \false^k)} (q',b_1,...,b_k).
          $

       \- Otherwise, do the following.
          \li
          \- Abort point:
             if there exists $i \in [k]$ such that $b_i = \false$ and $g$ contains $\mc i \neq \mc r_i$ or $\mc o \sim \mc r_i$,
             then abort.
             Because the register $\mc r_i$ is uninitalized ($b_i = \false$),
             we cannot know the valuation of $\mc i \neq \mc r_i$ or $\mc o \neq \mc r_i$.
             In contrast,
             if the guard $g$ contains $\mc i = \mc r_i$, we can assume that it holds and store $\mc i$ into $\mc r_i$
             (we cannot do this for $\mc o = \mc r_i$, because the automata do not allow for storing $\mc o$).

          \- Add to $\delta'$ the transition
             $
               (q,b_1,...,b_k) \lblTo{(l, g', a)} (q',b_1',...,b_k')
             $
             where for every $i \in [k]$:
             \li
   
             \- $b'_i=\true$ iff $b_i=\true$ or $g$ contains $\mc i = \mc r_i$.
   
             \- The action $a$ stores $\mc i$ into $\mc r_i$
                iff $g$ contains $\mc i = \mc r_i$ and $b_i=\false$.
   
             \- The guard $g'$ contains $\mc i \sim \mc r_i$
                iff $g$ contains $\mc i \sim \mc r_i$ and $b_i=\true$;
                similarly for $\mc o \sim \mc r_i$.
   
             \il

          \- Finally, we account for the inequality set $E$ and update $g'$ as follows.
             For every $(\mc r_i, \mc r_j) \in E$:
             if $b_i=\true$ and the action $a$ contains $\mc r_j = \mc i$,
             then add to $g'$ the expression $\mc i \neq r_i$.

             (Here we assume that the $A$-transition is not contradictory, namely, it is not the case that
              $\exists (\mc r_i, \mc r_j) \in E: b_i = \false \land b_j = \false \land (\mc i = \mc r_i) \in g \land (\mc i = \mc r_j) \in g$.
              Such transitions cannot be executed in $A$ and can be removed beforehand.)
          \il
       \il

       \- Note that the automaton $A'$ never compares $\mc i$ nor $\mc o$ with a register that was uninitialized.
          Therefore, the component $\mc d_0$ of $A'$ can be anything from $\mc D$.
    \il
    The automaton $A'$ has $|Q'| = |Q|\!\cdot\!2^{k}$,
    but the number of reachable states is $|Q|\cdot k$,
    because every $A'$-transition $(q,b_1,...,b_k) \lblTo{(l,g,a)} (q',b'_1,...,b'_k)$ satisfies $(b'_1,...,b'_k) \geq (b_1,...,b_k)$.

    An example of the conversion is in Figure~\ref{fig:reg-automaton-translated}.
    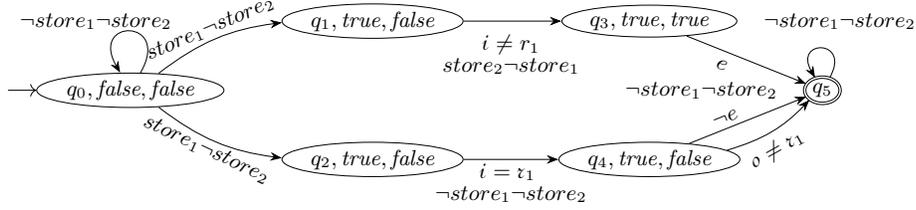
\begin{figure}[tb]
      \centering
      \scalebox{0.92}{\begin{tikzpicture}
	\begin{pgfonlayer}{nodelayer}
		\node [style=text ellipse, initial] (0) at (-3.5, 1) {$q_0, \false, \false$};
		\node [style=text ellipse] (1) at (0, 2) {$q_1, \true, \false$};
		\node [style=textual] (2) at (-4, 2) {$\neg store_1 \neg store_2$};
		\node [style=textual, rotate=23, yshift=-1.5mm, xshift=1mm] (3) at (-2.5, 2) {$store_1\neg store_2$};
		\node [style=text ellipse] (4) at (4, 2) {$q_3, \true, \true$};
		\node [double, style=text ellipse] (5) at (6.5, 1) {$q_5$};
		\node [style=text ellipse] (6) at (0, -0) {$q_2, \true, \false$};
		\node [style=text ellipse] (7) at (4, -0) {$q_4, \true, \false$};
		\node [xshift=1mm, yshift=-1mm, rotate=-25, style=textual] (8) at (-2.5, 0.25) {$store_1\neg store_2$};
		\node [style=textual] (9) at (2, 1.5) {\specialcellC{$\mc i \neq r_1$\vspace{-0.8mm}\\   $store_2 \neg store_1$}};
		\node [style=textual, yshift=-1mm] (10) at (2, -0.25) {\specialcellC{$\mc i = \mc r_1$  \vspace{-0.8mm}  \\  $\neg store_1 \neg store_2$}};
		\node [yshift=1mm, xshift=1mm, rotate=22, style=textual] (11) at (5, 0.5) {$\neg e$};
		\node [style=textual, rotate=-25, yshift=-1mm, xshift=1mm] (12) at (5, 1.5) {$e$};
		\node [xshift=1mm, yshift=-1mm, style=textual, rotate=29] (13) at (5.75, 0.25) {$\mc o \neq \mc r_1$};
		\node [style=textual] (14) at (4.75, 1) {$\neg store_1 \neg store_2$};
		\node [style=textual] (15) at (6.75, 2) {$\neg store_1 \neg store_2$};
	\end{pgfonlayer}
	\begin{pgfonlayer}{edgelayer}
		\draw [style=arrow, bend left=15, looseness=1.00] (0) to (1);
		\draw [style=arrow, in=120, out=60, loop] (0) to ();
		\draw [style=arrow, in=120, out=60, loop] (5) to ();
		\draw [style=arrow, bend right=15, looseness=1.00] (0) to (6);
		\draw [style=arrow] (6) to (7);
		\draw [style=arrow, bend right=15, looseness=1.00] (7) to (5);
		\draw [style=arrow] (1) to (4);
		\draw [style=arrow] (4) to (5);
		\draw [style=arrow] (7) to (5);
	\end{pgfonlayer}
\end{tikzpicture}}
      \caption{A 2-register automaton converted from the register-guessing automaton in Figure~\ref{fig:reg-guessing-automaton}.}
      \label{fig:reg-automaton-translated}
    \end{figure}

    \begin{obs}\label{obs:gue-reg-atm--reg-atm}
      Given a register-guessing automaton $A$.
      If conversion-2 succeeds and produces a register automaton $A'$,
      then $L(A) = L(A')$,
    \end{obs}

    \begin{proof}[Proof]
      We need to prove that
      $\forall w \in (2^P \x \mc D^2)^\omega: w \models A \Iff w \models A'$.

      Consider the direction $\Impl$.
      The acceptance $w \models A$ means that
      there exists $\mc R_0 \in \mc D^k$ and an accepting data-path $p$ starting in configuration $(q_0, \mc R_0)$
      and corresponding to $p$ the accepting Boolean path $p_\bbB$:
      \begin{align*}
             &p = (q_0, \mc R_0) \lblTo{(l_0, \mc i_0, \mc o_0)} (q_1, \mc R_0) \lblTo{(l_1, \mc i_1, \mc o_1)} ... \\
        &p_\bbB = q_0 \lblTo{(l_0, g_0)} q_1 \lblTo{(l_1, g_1)} ...
      \end{align*}
      (for every $j \in \bbN_0$:
       $l_j \in 2^P$,
       ${\mc i}_j \in \mc D$,
       ${\mc o}_j \in \mc D$, and
       $g_j \in G$).
      We build inductively the accepting data-path $p'$ of $A'$ (and corresponding to $p'$ the Boolean path $p'_\bbB$):
      \begin{align*}
        &p' = \big((q_0, B_0), \mc R'_0\big) \lblTo{(l_0, \mc i_0, \mc o_0, a'_0)} \big((q_1, B_1), \mc R'_1\big) \lblTo{(l_1, \mc i_1, \mc o_1, a'_1)} ... \\
        &p'_\bbB = (q_0, B_0) \lblTo{(l_0, g'_0, a'_0)} (q_1, B_1) \lblTo{(l_1, g'_1, a'_1)} ...
      \end{align*}
      ($\forall j \in \bbN_0$:
       $B_j \in \bbB^k$, $a'_j \in \bbB^k$, $\mc R_j \in \mc D^k$, and $\mc i_j$, $\mc o_j$, $l_j$ are as for $p$) as follows.
      \li
      \- The path $p'$ starts with $B_0 = \false^k$ and $\mc R_0 = \mc d_0^k$.
      \- The construction of $A'$ uniquely defines $a'_j$, $g'_j$, and $B_{j+1}$
         from $B_j$ and $g_j$.
      \- The value of $\mc R'_{j+1}$ is uniquely defined by $\mc R'_j$, $a'_j$, and $\mc i$.
      \- By the construction, every $g'_j$ is less restricting or equal to $g_j$,
         and we never compare $\mc i$ or $\mc o$ with uninitialized registers.
         Because $(\mc i_j, \mc o_j, \mc R_0) \models g_j$,
         we have that $(\mc i_j, \mc o_j, \mc R'_j) \models g'_j$, for every $j \in \bbN_0$.
         Hence, every transition of $p'$ is indeed a transition of $A'$, and $p'$ is indeed a path of $A'$.
      \il

      The direction $\Leftarrow$ is similar to the above.
      The data-path $p$ and corresponding to $p$ the Boolean path $p_\bbB$ of $A$
      are uniquely constructed from a given data-path $p'$ and corresponding to $p'$ the Boolean path $p'_\bbB$ of $A'$.
      When proving that $p$ is indeed a path of $A$,
      we use the property of $A'$ that it never compares $\mc i$ nor $\mc o$ with a register whose value was not written before.
    \end{proof}

    Combined together, the conversions give us the following.
    \begin{theorem}
      Given an $\exists$LTL(EQ) $\Phi = \exists \mc x_1,...,\mc x_k.cond.\varphi$.
      If conversion-1 and conversion-2 succeed and result in a register automaton $A$, then $L(\Phi) = L(A')$.
    \end{theorem}


\section{Conclusion} \label{sec:conclusion}
  \ak{mention Thomas's paper on games}
  \ak{mention Ranko's theorem in that paper that also uses abstrac classes}

  In this paper we introduced a sound and complete approach to synthesis of register transducers from specifications given as register automata.
  Although we focused on automata with the co-B\"uchi acceptance, others (e.g., parity) looks doable.
  The approach works (incompletely) for specifications given as quantified temporal logic formulas,
  by converting them into register automata.
  In particular,
  we investigated the two directions---richer automata and suitable temporal logic---raised by Ehlers et al.\cite[Sect.6]{ehlers-register-machines}.

  We are working on extending the approach to automata with guards that, in addition to $=$, have operators $>$, $+$,
  and on the question of decidability of the unbounded-but-finite synthesis problem that is open.
  It would be interesting to combine our approach with the approach to synthesis of reactive programs~\cite{2018arXiv180709047G}.
  It would also be interesting to do a synthesis case study, possibly for specifications with costs.

\ack{We thank Orna Kupferman for comments on the early draft and helpful discussions.
  This work was supported by the Austrian Science Fund (FWF) under the RiSE
  National Research Network (S11406) and by the Hebrew University.%
}

\vspace{-3.5mm}
\bibliographystyle{paper-includes/splncs03}
\bibliography{refs}

\appendix
\section{Change History}
\li
\- Aug 19, 2018: fixed typos in Section~\ref{sec:AV} (it incorrectly used $\ATV$ instead of $A_\bbB @ V_k$).
\- Aug 14, 2018: fixed a small bug in Section~\ref{sec:ATW}.
\- Aug 9, 2018: rewriting into using universal instead of non-deterministic automata.
   Added complexity result.
\- Aug 2, 2018: the extended version of the final version submitted to ATVA.
\il

\end{document}